\documentclass[graybox]{svmult}

\usepackage{amsmath,amssymb}
\usepackage{mathptmx}       
\usepackage{helvet}         
\usepackage{courier}        
\usepackage{type1cm}        
\usepackage{color}
\usepackage{makeidx}         
\usepackage{graphicx}        
\usepackage{multicol}        
\usepackage[bottom]{footmisc}
\usepackage{bm}

\newtheorem{idea memo}{Idea Memo}[section]
\newtheorem{notation}{Notation}[section]
\newtheorem{proposal}{Davies-Lewis proposal}

\newtheorem{dlocriterion}{Davies-Lewis-Ozawa criterion}

\makeindex             

\begin{document}

\title*{Measuring processes and the Heisenberg picture}
\author{Kazuya Okamura}
\institute{Kazuya Okamura \at 
Graduate School of Informatics, Nagoya University
Chikusa-ku, Nagoya 464-8601, Japan, \email{okamura@math.cm.is.nagoya-u.ac.jp}}
%
%
\maketitle

\abstract*{In this paper, we attempt to establish quantum measurement theory in the Heisenberg picture.
First, we review foundations of quantum measurement theory,
that is usually based on the Schr\"{o}dinger picture. The concept of instrument is introduced there.
Next, we define the concept of system of measurement correlations and that of measuring process.
The former is the exact counterpart of instrument in the (generalized) Heisenberg picture.
In quantum mechanical systems, we then show
a one-to-one correspondence between systems of measurement correlations and measuring processes
up to complete equivalence. This is nothing but a unitary dilation theorem of
systems of measurement correlations.
Furthermore, from the viewpoint of the statistical approach to quantum measurement theory,
we focus on the extendability of instruments to systems of measurement correlations.
It is shown that all completely positive (CP) instruments are extended into systems of measurement correlations.
Lastly, we study the approximate realizability of CP instruments
by measuring processes within arbitrarily given error limits.}

\abstract{
In this paper, we attempt to establish quantum measurement theory in the Heisenberg picture.
First, we review foundations of quantum measurement theory,
that is usually based on the Schr\"{o}dinger picture. The concept of instrument is introduced there.
Next, we define the concept of system of measurement correlations and that of measuring process.
The former is the exact counterpart of instrument in the (generalized) Heisenberg picture.
In quantum mechanical systems, we then show
a one-to-one correspondence between systems of measurement correlations and measuring processes
up to complete equivalence. This is nothing but a unitary dilation theorem of
systems of measurement correlations.
Furthermore, from the viewpoint of the statistical approach to quantum measurement theory,
we focus on the extendability of instruments to systems of measurement correlations.
It is shown that all completely positive (CP) instruments are extended into systems of measurement correlations.
Lastly, we study the approximate realizability of CP instruments
by measuring processes within arbitrarily given error limits.}

\section{Introduction}
\label{sec:1}
In this paper, we mathematically investigate measuring processes in the Heisenberg picture.
We aim to extend the framework of quantum measurement theory and
to apply the method in this paper not only to quantum systems of finite degrees of freedom
but also to those with infinite degrees of freedom.

It is well-known that correlation functions are essential for the description of systems
in quantum theory and in quantum probability theory.
Typical examples are Wightman functions in (axiomatic) quantum field theory \cite{SW00} and several
(algebraic, noncommutative) independence in quantum probability theory
(see \cite{HoraObata07,Muraki13} and references therein),
which are characterized by behaviors of correlation functions.
In the famous paper by Accardi, Frigerio and Lewis \cite{AFL82},
general classes of quantum stochastic processes including quantum Markov processes
were characterized in terms of correlation functions.
The main result of \cite{AFL82} is a noncommutative version of Kolmogorov's theorem
stating that every quantum stochastic process can be reconstructed
from a family of correlation functions up to equivalence.
The proof of this result is made by the efficient use of
positive-definiteness of a family of correlation functions.
Later Belavkin \cite{Belavkin85} formulated the theory of
operator-valued correlation functions, which is more flexible than the original formulation in \cite{AFL82}
and gives an opportunity for reconsidering the standard formulation of qunatum theory.
And he extended the main result of \cite{AFL82}.
We apply Belavkin's theory, with some modifications, to a systematic characterization of
measurement correlations herein.


Measurements are described by the notion of instrument introduced Davies and Lewis \cite{DL70}.
An instrument $\mathcal{I}$ for $(\mathcal{M},S)$ is defined as a $P(\mathcal{M}_\ast)$-valued measure
$\mathcal{F}\ni\Delta\mapsto\mathcal{I}(\Delta)\in P(\mathcal{M}_\ast)$,
where $\mathcal{M}$ is a von Neumann algebra with predual $\mathcal{M}_\ast$,
$P(\mathcal{M}_\ast)$ is the set of positive linear map of $\mathcal{M}$
and $(S,\mathcal{M})$ is a measurable space.
The statistical description of measurements in terms of instruments can be regarded as a kind of
quantum dynamical process based on the so-called Schr\"{o}dinger picture.
As widely accepted, the Schr\"{o}dinger picture stands
on describing states as time-dependent variables and observables as constants with respect to time, i.e.,
time-independent variables while we treat
states as constants with respect to time and observables as time-dependent variables the Heisenberg picture.
To the author's knowledge, the Schr\"{o}dinger picture is
matched with an operational approach to quantum theory
concerning probability distributions of observables and of output variables of apparatuses
\cite{DL70,Davies76,BGL95,Ozawa97}.
On the other hand, no systematic treatment of measurements
in the Heisenberg picture, which can compare with the theory of instruments, has been investigated.
In contrast to the Schr\"{o}dinger picture, the Heisenberg picture focuses on
dynamical changes of observables and can naturally treat correlation functions
of observables at different times,
so that enables us to examine the dynamical nature of the system under consideration itself in detail.
The Heisenberg picture is better than the Schr\"{o}dinger picture at this point.
Therefore, inspired by the previous investigations on 
quantum stochastic processes and correlation functions \cite{AFL82,Belavkin85},
we define the concept of system of measurement correlations.
This is the exact counterpart of instrument in a ``generalized" Heisenberg picture and
defined as a family of ``operator-valued" multilinear maps satisfying ``positive-definiteness",
``$\sigma$-additivity" and other conditions.
An instrument induced by a system of measurement correlations is always completely positive.
In addition, we redefine measuring process (Definition \ref{MP2}) in order that
it becomes consistent with the definition of system of measurement correlations.
In the quantum mechanical case, we show that every system of measurement correlations is
defined by a measuring process.
It is, however, difficult to extend this result to general von Neumann algebras.
Therefore, we develop another aspect of measurements,
which is deeply analyzed for the first time in this paper.
From the statistical viewpoint as the starting point of quantum measurement theory,
we discuss the extendability of CP instruments to
systems of measurement correlations and the realizability of CP instruments by measuring processes.
In physically relevant cases, we show that both are possible within arbitrary accuracy.

The purpose of this paper is to mathematically
develop the unitary dilation theory of systems of measurement correlations and of CP instruments.
Dilation theory is one of main topics in functional analysis and enables us to apply
representation theory and harmonic analysis to operators or to operator algebras.
Especially, the unitary dilation theory of contractions on Hilbert space \cite{NFBK10,Halmos82}
and the dilation theory of completely positive maps \cite{Arveson69,Stine55}
have been studied in many investigations
(see \cite{NFBK10,Halmos82,Paulsen02,Arveson69,Stine55,
Lance95,Kraus,Kraus2,Skeide00,Skeide01} and references therein).
A representation theorem of CP instruments on the set ${\bf B}(\mathcal{H})$ of bounded operators
on a Hilbert space $\mathcal{H}$ \cite[Theorem 5.1]{Oz84} (Theorem \ref{CPIMP}) follows from these results,
which shows the existence of unitary dilations of CP instruments.
The proof of this theorem is based on the theory of CP-measure space \cite{O14,OO16}.
In the case of CP instruments,
a unitary dilation of a CP instrument is nothing but a measuring process which realizes it.
We generalize this representation theorem to systems of
measurement correlations defined on ${\bf B}(\mathcal{H})$ in terms of Kolmogorov's theorem.
It should be remarked that CP instruments defined on general von Neumann algebras do not always admit unitary dilations
(see Examples \ref{NoNEP1} and \ref{NoNEP2}).
Next, we consider the extendability of CP instruments to systems of measurement correlations.
It will be shown that all CP instruments can be extended into systems of measurement correlations.
Furthermore, we show that
every CP instrument defined on general von Neumann algebras can be approximated
by measuring processes within arbitrarily given error limits $\varepsilon>0$.
If von Neumann algebras are injective or injective factors, measuring processes
approximating a CP instrument can be chosen to be faithful or inner, respectively.

Preliminaries are given in Section \ref{sec:2}; Foundations of quantum measurement theory,
kernels and their Kolmogorov decompositions are explained.
We introduce a system of measurement correlations and prove a representation theorem
of systems of measurement correlations in Section \ref{sec:3}.
In Section \ref{sec:4}, we define measuring processes and their complete equivalence,
and in the case of ${\bf B}(\mathcal{H})$ we show a unitary dilation theorem of
systems of measurement correlations establishing
a one-to-one correspondence between systems of measurement correlations and
complete equivalence classes of measuring processes.
In Section \ref{sec:5}, we discuss a generalization of the main result in Section \ref{sec:4}
to arbitrary von Neumann algebras, and the extendability of CP instruments to
systems of measurement correlations.
We show that for any CP instruments
there always exists a systems of measurement correlations which defines a given CP instrument.
In Section \ref{sec:6}, we explore the existence of measuring processes which approximately
 realizes a given CP instrument.
We show several approximate realization theorems of CP instruments by measuring processes.
\section{Preliminaries}
\label{sec:2}
In this paper, we assume that von Neumann algebras $\mathcal{M}$ are $\sigma$-finite.
However, only in the case of $\mathcal{M}={\bf B}(\mathcal{H})$, the von Neumann algebra of
bounded operators on a Hilbert space $\mathcal{H}$, this assumption is ignored.
\subsection{Foundations of quantum measurement theory}
We introduce foundations of quantum measurement theory herein.
To precisely understand the theory of quantum measurement and its mathematics,
the most important thing is to know how measurements physically realizable in experimental settings
are described as physical processes consistent with statistical characterization of measurements.
We refer the reader to \cite{Ozawa04,14MFQ-,OO16} for detailed introductory expositions
of quantum measurement theory.

The history of quantum measurement theory is long as much as those of quantum theory,
but the modern theory of quantum measurement began with the mathematical study of
the notion of instruments introduced by Davies and Lewis \cite{DL70}. They proposed that
we should abandon \textit{the repeatability hypothesis} \cite{vonNeumann,DL70,Ozawa15} as a general principle
and employ an operational approach to quantum measurement,
which is based on the mathematical description of measurements
in terms of instruments defined as follows.
Let $\mathbf{S}$ be a system whose observables and states are
described by self-adjoint operators affiliated to
a von Neumann algebra $\mathcal{M}$ on a Hilbert space $\mathcal{H}$
and by normal states on $\mathcal{M}$, respectively.
$\mathcal{M}_\ast$ denotes the predual of $\mathcal{M}$,
i.e., the set of ultraweakly continuous linear functionals on $\mathcal{M}$,
$\mathcal{S}_n(\mathcal{M})$ does that of normal states on $\mathcal{M}$
and $P(\mathcal{M}_\ast)$ does that of positive linear maps on $\mathcal{M}_\ast$.

\begin{definition}[Instrument, Davies-Lewis \text{\cite[p.243, ll.21--26]{DL70}}]
Let $\mathcal{M}$ be a von Neumann algebra on a Hilbert space $\mathcal{H}$
and $(S,\mathcal{F})$ a measurable space. 
A map $\mathcal{I}:\mathcal{F}\rightarrow P(\mathcal{M}_\ast)$ is called an instrument for $(\mathcal{M},S)$
if it satisfies the following conditions:\\
$(1)$ $\langle\mathcal{I}(S)\rho,1\rangle=\langle\rho,1\rangle$ for all $\rho\in\mathcal{M}_\ast$,
where $\langle \cdot,\cdot \rangle$ denotes the duality pairing of $\mathcal{M}_\ast$ and $\mathcal{M}$;\\
$(2)$ For every $M\in\mathcal{M}$, $\rho\in\mathcal{M}_\ast$ and mutually disjoint sequence $\{\Delta_j\}$
of $\mathcal{F}$,
\begin{equation}
\langle \mathcal{I}(\cup_j\Delta_j)\rho, M \rangle=\sum_j \langle \mathcal{I}(\Delta_j)\rho, M \rangle.
\end{equation}

We define the dual map $\mathcal{I}^\ast$ of an instrument $\mathcal{I}$ by
$\langle \rho,\mathcal{I}^\ast(\Delta)M \rangle=\langle \mathcal{I}(\Delta)\rho,M \rangle$
and use the notation $\mathcal{I}(M,\Delta)=\mathcal{I}^\ast(\Delta)M$ for all
$\Delta\in\mathcal{F}$ and $M\in\mathcal{M}$.
\end{definition}
It is obvious, by the definition, that every instrument describes
the weighted state changes caused by the measurement.
The dual map $\mathcal{I}:\mathcal{M}\times\mathcal{F}\rightarrow \mathcal{M}$
of an instrument $\mathcal{I}$ for $(\mathcal{M},S)$
is characterized by the following conditions:\\
$(i)$ For every $\Delta\in\mathcal{F}$, the map $\mathcal{M}\ni M\mapsto \mathcal{I}(M,\Delta)\in\mathcal{M}$
is normal positive linear map of $\mathcal{M}$;\\
$(ii)$ $\mathcal{I}(1,S)=1$;\\
$(iii)$ For every $M\in\mathcal{M}$, $\rho\in\mathcal{M}_\ast$ and mutually disjoint sequence $\{\Delta_j\}$
of $\mathcal{F}$,
\begin{equation}
\langle \rho, \mathcal{I}(M,\cup_j\Delta_j) \rangle=\sum_j \langle \rho, \mathcal{I}(M,\Delta_j) \rangle.
\end{equation}
Since every map $\mathcal{I}:\mathcal{M}\times\mathcal{F}\rightarrow \mathcal{M}$ satisfying the above conditions
is always the dual map of an instrument for $(\mathcal{M},S)$,
we also call the map $\mathcal{I}$ an instrument for $(\mathcal{M},S)$.

Davies and Lewis claimed that
experimentally and statistically accessible ingredients via measurements by a given measuring apparatus
should be specified by instruments as follows:
\begin{proposal}
For every apparatus ${\bf A}({\bf x})$ measuring ${\bf S}$,
where ${\bf x}$ is the output variable of ${\bf A}({\bf x})$
taking values in a measurable space $(S,\mathcal{F})$,
there always exists an instrument $\mathcal{I}$ for $(\mathcal{M},S)$
corresponding to ${\bf A}({\bf x})$ in the following sense.
For every input state $\rho$, the probability distribution
$\mathrm{Pr}\{{\bf x}\in\Delta\Vert\rho\}$ of ${\bf x}$ in $\rho$ is given by
\begin{equation}
\mathrm{Pr}\{{\bf x}\in\Delta\Vert\rho\} = \Vert\mathcal{I}(\Delta)\rho\Vert
=\langle \mathcal{I}(\Delta)\rho, 1\rangle
\end{equation}
for all $\Delta\in\mathcal{F}$,
and the state $\rho_{\{{\bf x}\in\Delta\}}$ after the measurement under the condition that
$\rho$ is the prepared state and the outcome ${\bf x}\in\Delta$ 
is given by
\begin{equation}
 \rho_{\{{\bf x}\in\Delta\}}= \dfrac{\mathcal{I}(\Delta)\rho}{\Vert\mathcal{I}(\Delta)\rho\Vert}
\end{equation}
if $\mathrm{Pr}\{{\bf x}\in\Delta\Vert\rho\}>0$, and $\rho_{\{{\bf x}\in\Delta\}}$ is indefinite if
$\mathrm{Pr}\{{\bf x}\in\Delta\Vert\rho\}=0$.
\end{proposal}
Although this proposal is very general, it was not evident at that time
that how this is related to the standard formulation of quantum theory.
In the 1980s, Ozawa \cite{Ozawa83,Oz84} introduced
both completely positive (CP) instruments and measuring processes.
Following this investigation, the standpoint of the above proposal in quantum mechanics was settled
and the circumstances changed at all.
An instrument $\mathcal{I}$ for $(\mathcal{M},S)$ is said to be \textit{completely positive} if 
$\mathcal{I}(\Delta)$ (or $\mathcal{I}(\cdot,\Delta)$, equivalently)
is completely positive for every $\Delta\in\mathcal{F}$.
$\mathrm{CPInst}(\mathcal{M},S)$ denotes the set of CP instruments for $(\mathcal{M},S)$.
The notion of measuring process is defined as
a quantum mechanical modeling of an apparatus as a physical system, of the meter of the apparatus, and
of the measuring interaction between the system and the apparatus.
Let  $\mathcal{H}$ and $\mathcal{K}$ be Hilbert spaces.
${\bf B}(\mathcal{H})$ denotes the set of bounded linear operators on $\mathcal{H}$
and ${\bf B}(\mathcal{H},\mathcal{K})$ does the set of bounded linear operators
of $\mathcal{H}$ to $\mathcal{K}$.
Let $\mathcal{M}$ and $\mathcal{N}$ be von Neumann algebras on $\mathcal{H}$ and $\mathcal{K}$,
respectively. $\mathcal{M}\overline{\otimes}\mathcal{N}$ denotes
the W$^\ast$-tensor product of $\mathcal{M}$ and $\mathcal{N}$.
For every $\sigma\in\mathcal{N}_\ast$, a linear map
$\mathrm{id}\otimes\sigma:\mathcal{M}\overline{\otimes}\mathcal{N}\rightarrow\mathcal{M}$ is defined by
$\langle \rho,(\mathrm{id}\otimes\sigma)X \rangle=\langle \rho\otimes\sigma, X\rangle$
for all $X\in\mathcal{M}\overline{\otimes}\mathcal{N}$ and $\rho\in\mathcal{M}_\ast$.
The following is the mathematical definition of measuring processes:
\begin{definition}[Measuring process \text{\cite[Definition 3.2]{OO16}}]\label{MP1}
A measuring process for $(\mathcal{M},S)$ is a 4-tuple $\mathbb{M}=(\mathcal{K},\sigma,E,U)$
of a Hilbert space $\mathcal{K}$, a normal state $\sigma$ on ${\bf B}(\mathcal{K})$,
a PVM $E:\mathcal{F}\rightarrow{\bf B}(\mathcal{K})$ and a unitary operator $U$ on
$\mathcal{H}\otimes\mathcal{K}$ satisfying
\begin{equation}
(\mathrm{id}\otimes \sigma)[U^\ast(M\otimes E(\Delta))U]\in \mathcal{M}
\end{equation}
for all $M\in\mathcal{M}$ and $\Delta\in\mathcal{F}$.
\end{definition}
Let $\mathbb{M}=(\mathcal{K},\sigma,E,U)$ be a measuring process for $(\mathcal{M},S)$.
Then a CP instrument $\mathcal{I}_\mathbb{M}$ for $(\mathcal{M},S)$ is defined by
\begin{equation}
\mathcal{I}_\mathbb{M}(M,\Delta)=(\mathrm{id}\otimes \sigma)[U^\ast(M\otimes E(\Delta))U]
\end{equation}
for $M\in\mathcal{M}$ and $\Delta\in\mathcal{F}$.
The most important example of meauring processes is a von Neumann model
$(L^2(\mathbb{R}),\omega_\alpha,E^Q,e^{-iA\otimes P/\hbar})$ of measurement
of an observable $A$, a self-adjoint operator affiliated with $\mathcal{M}$, where
$\alpha$ is a unit vector of $L^2(\mathbb{R})$,
$\omega_\alpha$ is defined by $\omega_\alpha(M)=\langle\alpha|M\alpha\rangle$
for all $M\in {\bf B}(L^2(\mathbb{R}))$,
 and $Q=\int_\mathbb{R} q\;dE^Q(q)$ and $P$ are self-adjoint operators defined on
dense linear subspaces of $L^2(\mathbb{R})$ such that $\overline{[Q,P]}=i\hbar 1$.
Here, $E^X$ denotes the spectral measure of a self-adjoint operator $X$ densely defined on a Hilbert space.
Quantum mechanical modeling of appratuses began with this model \cite{vonNeumann,Oz93}.

Two measuring processes are \textit{statistically equivalent} if they define an identical instrument.
As seen above, a measuring process $\mathbb{M}$ for $(\mathcal{M},S)$ defines
a CP instrument $\mathcal{I}_\mathbb{M}$ for $(\mathcal{M},S)$.
In the case of $\mathcal{M}={\bf B}(\mathcal{H})$, the following theorem,
a unitary dilation theorem of CP instruments for $({\bf B}(\mathcal{H}),S)$, is known to hold.
\begin{theorem}[\text{\cite{Ozawa83}, \cite[Theorem 5.1]{Oz84}, \cite[Theorem 3.6]{OO16}}]\label{CPIMP}
For every CP instrument $\mathcal{I}$ for $({\bf B}(\mathcal{H}),S)$,
there uniquely exists a statistical equivalence class of measuring processes
$\mathbb{M}=(\mathcal{K},\sigma,E,U)$ for $({\bf B}(\mathcal{H}),S)$ such that
$\mathcal{I}(M,\Delta)=\mathcal{I}_\mathbb{M}(M,\Delta)$
for all $M\in{\bf B}(\mathcal{H})$ and $\Delta\in\mathcal{F}$. Conversely,
every statistical equivalence class of measuring processes for $({\bf B}(\mathcal{H}),S)$
defines a unique CP instrument $\mathcal{I}$ for $({\bf B}(\mathcal{H}),S)$.
\end{theorem}
A generalization of this theorem to arbitrary von Neumann algebras is shown to hold
not for all CP instruments but for those with the normal extension property (NEP) in \cite{OO16}
(see Theorem \ref{THMNEP}).
Let $(S,\mathcal{F},\mu)$ be a measure space.
$\mathcal{L}(S,\mu)$ denotes the $^\ast$-algebra of complex-valued $\mu$-measurable functions on $S$.
A $\mu$-measurable function $f$ is said to be $\mu$-negligible if $f(s)=0$ for $\mu$-a.e. $s\in S$.
$\mathcal{N}(S,\mu)$ denotes the set of $\mu$-negligible functions on $S$
and $M^\infty(S,\mu)$ does the $^\ast$-algebra of bounded $\mu$-measurable functions on $S$.
It is obvious that $M^\infty(S,\mu)\subset\mathcal{L}(S,\mu)$ as $^\ast$-algebra.
For any $1\leq p<\infty$, $L^p(S,\mu)$ denotes the Banach space of
 $p$-integrable functions on $S$ with respect to $\mu$ modulo the $\mu$-negligible functions.
$[f]$ denotes the $\mu$-negligible equivalence class of $f\in\mathcal{L}(S,\mu)$
and $L^\infty(S,\mu)$ does the commutative von Neumann algebra on $L^2(S,\mu)$.
$L^\infty(S,\mathcal{I})$ denotes the W$^\ast$-algebra of essentially bounded
$\mathcal{I}$-measurable functions on $S$ modulo the $\mathcal{I}$-negligible functions.

\begin{definition}[Normal extension property \text{\cite[Definition 3.3]{OO16}}]
Let $\mathcal{M}$ be a von Neumann algebra on a Hilbert space $\mathcal{H}$
and $(S,\mathcal{F})$ a measurable space.
Let $\mathcal{I}$ be a CP instrument for $(\mathcal{M},S)$ and
$\Psi_\mathcal{I}:\mathcal{M}\otimes_{\mathrm{min}}L^\infty(S,\mathcal{I})\rightarrow\mathcal{M}$
the corresponding unital binormal CP map, i.e., $\Psi_\mathcal{I}$ is
normal on $\mathcal{M}$ and $L^\infty(S,\mathcal{I})$
and satisfies $\Psi_\mathcal{I}(M\otimes [\chi_\Delta])=\mathcal{I}(M,\Delta)$
for all $M\in\mathcal{M}$ and $\Delta\in\mathcal{F}$
\cite[Proposition 3.3]{OO16}.
$\mathcal{I}$ is said to have the normal extension property (NEP)
if there exists a unital normal CP map
$\widetilde{\Psi_\mathcal{I}}:\mathcal{M}\overline{\otimes}L^\infty(S,\mathcal{I})\rightarrow\mathcal{M}$
such that $\widetilde{\Psi_\mathcal{I}}|_{\mathcal{M}\otimes_{\mathrm{min}}L^\infty(S,\mathcal{I})}
=\Psi_\mathcal{I}$. $\mathrm{CPInst}_{\mathrm{NE}}(\mathcal{M},S)$ denotes
the set of CP instruments for $(\mathcal{M},S)$ with the NEP.
\end{definition}

We then have the following theorem, a generalization of Theorem \ref{CPIMP}.

\begin{theorem}[\text{\cite[Theorem 3.4]{OO16}}] \label{THMNEP}
For a CP instrument $\mathcal{I}$ for $(\mathcal{M},S)$, the following conditions are equivalent:\\
$(i)$ $\mathcal{I}$ has the NEP.\\
$(ii)$ There exists CP instrument $\mathcal{I}$ for $({\bf B}(\mathcal{H}),S)$ such that
$\widetilde{\mathcal{I}}(M,\Delta)=\mathcal{I}(M,\Delta)$ for all $M\in\mathcal{M}$ and $\Delta\in\mathcal{F}$.\\
$(iii)$ There exists a measuring process $\mathbb{M}=(\mathcal{K},\sigma,E,U)$ for $(\mathcal{M},S)$ such that
$\mathcal{I}(M,\Delta)=\mathcal{I}_\mathbb{M}(M,\Delta)$ for all $M\in\mathcal{M}$ and $\Delta\in\mathcal{F}$.
\end{theorem}

It is also shown that all CP instruments defined on a von Neumann algebra $\mathcal{M}$
have the NEP, i.e., $\mathrm{CPInst}(\mathcal{M},S)=\mathrm{CPInst}_{\mathrm{NE}}(\mathcal{M},S)$
if $\mathcal{M}$ is atomic \cite[Theorem 4.1]{OO16}.
We should remember the famous fact that
a von Neumann algebra $\mathcal{M}$ on a Hilbert space $\mathcal{H}$ is atomic
if and only if there exists a normal conditional expectation
$\mathcal{E}:{\bf B}(\mathcal{H})\rightarrow\mathcal{M}$ \cite[Chapter V, Section 2, Excercise 8]{T79}.
Then the following question naturally arises.

\begin{question}\label{Q3}
Let $\mathcal{M}$ be a non-atomic von Neumann algebra on a Hilbert space
$\mathcal{H}$ and $(S,\mathcal{F})$ a measurable space.
Are there CP instruments for $(\mathcal{M},S)$ without the NEP?
For any CP instrument $\mathcal{I}$ for $(\mathcal{M},S)$, does there exist
a measuring process $\mathbb{M}$ for $(\mathcal{M},S)$
which realizes $\mathcal{I}$ within arbitrarily given error limits $\varepsilon>0$?
\end{question}

A CP instrument $\mathcal{I}$ for $(\mathcal{M},S)$ is said to have
\textit{the approximately normal extension property} (ANEP)
if there is a net $\{\mathcal{I}_\alpha\}$ of CP instruments with the NEP
such that $\mathcal{I}_\alpha(M,\Delta)$ is ultraweakly converges to $\mathcal{I}(M,\Delta)$
for all $M\in\mathcal{M}$ and $\Delta\in\mathcal{F}$. 
$\mathrm{CPInst}_{\mathrm{AN}}(\mathcal{M},S)$ denotes the set of 
CP instruments for $(\mathcal{M},S)$ with the ANEP.

Contrary to physicists' expectations,
Question \ref{Q3} was positively resolved in \cite[Section V]{OO16} for
non-atomic but injective von Neumann algebras.
\begin{definition}[\text{\cite[Definition 5.3]{OO16}}]
$(1)$ An instrument $\mathcal{I}$ for $(\mathcal{M},S)$ is called repeatable if
$\mathcal{I}(\Delta_2)\mathcal{I}(\Delta_1)=\mathcal{I}(\Delta_2\cap\Delta_1)$
for all $\Delta_1,\Delta_2\in\mathcal{F}$.\\
$(2)$ An instrument $\mathcal{I}$ for $(\mathcal{M},S)$ is called weakly repeatable if
$\mathcal{I}(\mathcal{I}(1,\Delta_2),\Delta_1)=\mathcal{I}(1,\Delta_2\cap\Delta_1)$
for all $\Delta_1,\Delta_2\in\mathcal{F}$.\\
$(3)$ An instrument $\mathcal{I}$ for $(\mathcal{M},S)$ is called discrete if
there exist a countable subset $S_0$ of $S$ and a map $T:S_0\rightarrow P(\mathcal{M}_\ast)$
such that
\begin{equation}
\mathcal{I}(\Delta)=\sum_{s\in\Delta} T(s)
\end{equation}
for all $\Delta\in\mathcal{F}$.
\end{definition}
\begin{proposition}[\text{\cite[Proposition 5.9]{OO16}}]\label{DCPinst1}
Let $\mathcal{M}$ be a von Neumann algebra on a Hilbert space
$\mathcal{H}$ and $(S,\mathcal{F})$ a measurable space.
Every discrete CP instrument $\mathcal{I}$ for $(\mathcal{M},S)$ has the NEP.
\end{proposition}
\begin{theorem}[\text{\cite[Theorem 5.10]{OO16}}] \label{DNEP1}
Let $\mathcal{M}$ be a von Neumann algebra on a Hilbert space
$\mathcal{H}$ and $(S,\mathcal{F})$ a standard Borel space.
A weakly repeatable CP instrument $\mathcal{I}$ for $(\mathcal{M},S)$ is discrete
if and only if it has the NEP.
\end{theorem}

In the case where $\mathcal{M}$ is non-atomic,
there exist CP instruments for $(\mathcal{M},S)$ without the NEP.
The following two CP instruments are such examples.
\begin{example}[\textrm{\cite[pp.~292--293]{Oz85}, \cite[Example 5.1]{OO16}}]\label{NoNEP1}
Let $m$ be Lebesgue measure on $[0,1]$. A CP instrument $\mathcal{I}_m$ for
$(L^{\infty}([0,1],m), [0,1])$ is defined by $\mathcal{I}_m(f,\Delta)=[\chi_\Delta] f$
for all $\Delta\in\mathcal{B}([0,1])$ and $f\in L^{\infty}([0,1],m)$.
\end{example}

A von Neumann algebra $\mathcal{M}$ is said to be
approximately finite-dimensional (AFD) if there is an increasing net
$\{\mathcal{M}_\alpha\}_{\alpha\in A}$ of finite-dimensional von Neumann subalgebras of $\mathcal{M}$
such that
\begin{equation}
\mathcal{M}=\overline{\bigcup_{\alpha\in A} \mathcal{M}_\alpha}^{uw}.
\end{equation}
\begin{example}[\textrm{\cite[Example 5.2]{OO16}}]\label{NoNEP2}
Let $\mathcal{M}$ be an AFD von Neumann algebra of type $\mathrm{II}_1$
on a separable Hilbert space $\mathcal{H}$.
Let $A=\int_\mathbb{R}a\;dE^A(a)$
be a self-adjoint operator with continuous spectrum affiliated with $\mathcal{M}$
and $\mathcal{E}$ a (normal) conditional expectation of $\mathcal{M}$ onto
$\{A\}^\prime\cap\mathcal{M}$ (the existence of $\mathcal{E}$ was first found by \cite[Theorem 1]{U54}),
where $\{A\}^\prime=\{E^{A}(\Delta)\;|\; \Delta\in \mathcal{B}(\mathbb{R})\}^\prime$.
A CP instrument $\mathcal{I}_A$ for $(\mathcal{M},\mathbb{R})$ is defined by 
\begin{equation}\label{repeatable}
\mathcal{I}_A(M,\Delta)=\mathcal{E}(M)E^A(\Delta)
\end{equation}
for all $M\in\mathcal{M}$ and $\Delta\in\mathcal{B}(\mathbb{R})$.
\end{example}
By Theorem \ref{DNEP1},
the weak repeatability and the non-discreteness of $\mathcal{I}_m$ and $\mathcal{I}_A$
imply the non-existence of measuring processes which define them.
These examples are very important for the dilation theory of CP maps
since they revealed the existence of families of CP maps which do not admit unitary dilations.

The following theorem
holds for general $\sigma$-finite von Neumann algebras
without assuming any other conditions.
\begin{theorem} \label{GVN1}
Let $\mathcal{M}$ be a von Neumann algebra on a Hilbert space $\mathcal{H}$
and $(S,\mathcal{F})$ a measurable space.
For every CP instrument $\mathcal{I}$ for $(\mathcal{M},S)$,
$n\in\mathbb{N}$, $\rho_1,\cdots,\rho_n\in\mathcal{S}_n(\mathcal{M})$,
$M_1,\cdots,M_n\in\mathcal{M}$ and $\Delta_1,\cdots,\Delta_n\in\mathcal{F}$,
there exists a measuring process $\mathbb{M}=(\mathcal{K},\sigma,E,U)$ for $(\mathcal{M},S)$ such that
\begin{equation}
\langle\rho_j,\mathcal{I}(M_j,\Delta_j)\rangle
=\langle\rho_j, \mathcal{I}_\mathbb{M}(M_j,\Delta_j)\rangle
\end{equation}
for all $j=1,\cdots,n$.
\end{theorem}
\begin{proof}
Let $n\in\mathbb{N}$, $\rho_1,\cdots,\rho_n\in\mathcal{S}_n(\mathcal{M})$,
$M_1,\cdots,M_n\in\mathcal{M}\backslash \{0\}$
and $\Delta_1,\cdots,\Delta_n\in\mathcal{F}\backslash \{\emptyset\}$.
Let $\mathcal{F}^\prime$ be a $\sigma$-subfield of $\mathcal{F}$ generated by $\Delta_1,\cdots,\Delta_n,S$.
Let $\{\Gamma_i\}_{i=1}^m \subset \mathcal{F}^\prime\backslash\{\emptyset\}$ be a maximal partition of $S$, i.e.,
$\{\Gamma_i\}_{i=1}^m$ satisfies the following conditions:\\
$(1)$ For every $i=1,\cdots,m$, if $\Delta\in\mathcal{F}^\prime$ satisfies $\Delta\subset\Gamma_i$,
then $\Delta$ is $\Gamma_i$ or $\emptyset$;\\
$(2)$ $\cup_{i=1}^m \Gamma_i=S$;\\
$(3)$ $\Gamma_i \cap\Gamma_j=\emptyset$ if $i\neq j$.\\
We fix $s_1,\cdots.s_m\in S$ such that $s_i\in\Gamma_i$ for all $i=1,\cdots,m$.
We define a discrete CP instrument $\mathcal{I}^\prime$ for $(\mathcal{M},S)$ by
\begin{equation}
\mathcal{I}^\prime(M,\Delta)=\sum_{j=1}^m\delta_{s_j}(\Delta)\mathcal{I}(M,\Gamma_j)
\end{equation}
for all $M\in\mathcal{M}$ and $\Delta\in\mathcal{F}$.
It is then obvious that $\mathcal{I}^\prime$ satisfies
\begin{equation}
\langle\rho_j,\mathcal{I}(M_j,\Delta_j)\rangle
=\langle\rho_j, \mathcal{I}^\prime(M_j,\Delta_j)\rangle
\end{equation}
for all $j=1,\cdots,n$.
By Proposition \ref{DCPinst1}, there exists a measuring process
$\mathbb{M}=(\mathcal{K},\sigma,E,U)$ for $(\mathcal{M},S)$ such that
$\mathcal{I}^\prime(M,\Delta)=\mathcal{I}_\mathbb{M}(M,\Delta)$ 
for all $M\in\mathcal{M}$ and $\Delta\in\mathcal{F}$. The proof is complete.
\end{proof}
\begin{corollary}
Let $\mathcal{M}$ be a von Neumann algebra on a Hilbert space $\mathcal{H}$
and $(S,\mathcal{F})$ a measurable space. Then we have
\begin{equation}
\mathrm{CPInst}_{\mathrm{AN}}(\mathcal{M},S)= \mathrm{CPInst}(\mathcal{M},S).
\end{equation}
\end{corollary}

In the case where $\mathcal{M}$ is injective, 
the result stronger than Theorem \ref{GVN1} is shown in \cite[Theorem 4.2]{OO16}:
for every CP instrument $\mathcal{I}$ for $(\mathcal{M},S)$,
$\varepsilon>0$, $n\in\mathbb{N}$, $\{\rho_i\}_{i=1}^n\subset\mathcal{S}_n(\mathcal{M})$,
$\{\Delta_i\}_{i=1}^n\subset\mathcal{F}$ and $\{M_i\}_{i=1}^n\subset\mathcal{M}$,
there exists a measuring process $\mathbb{M}$ for $(\mathcal{M},S)$ such that
\begin{equation}
|\langle \mathcal{I}(\Delta_i)\rho_i, M_i\rangle-
\langle \mathcal{I}_\mathbb{M}(\Delta_i)\rho_i, M_i\rangle|<\varepsilon
\end{equation}
for all $i=1,2, \cdots,n$, and that $\mathcal{I}(1,\Delta)=\mathcal{I}_\mathbb{M}(1,\Delta)$
for all $\Delta\in\mathcal{F}$.
In physically relevant cases, it is known that
every von Neumann algebra $\mathcal{M}$ describing the observable algebra
of a quantum system acts on a separable Hilbert space and is AFD.
For example, it is shown in \cite{BDF87} that
von Neumann algebras of local observables in quantum field theory are AFD
and acts on a separable Hilbert space under natural postulates, e.g., the Wightman axioms,
the nuclearity condition and the asymptotic scale invariance.
For every von Neumann algebra $\mathcal{M}$ on a separable Hilbert space (or with separable dual, equivalently),
$\mathcal{M}$ is AFD if and only if it is injective, furthermore, if and only if it is amenable \cite{Connes,T02}.
Hence the assumption of the injectivity for von Neumann algebras is very natural.

In quantum mechanics, complete positivity of instruments
is physically justified in \cite{Ozawa04,14MFQ-} by considering a natural extendability, 
called \textit{the trivial extendability}, of an instrument $\mathcal{I}$ on the system ${\bf S}$ to
that $\mathcal{I}^\prime$ on the composite system ${\bf S}+{\bf S^\prime}$ containing
the original one ${\bf S}$, where ${\bf S^\prime}$ is an arbitrary system
not interacting with ${\bf S}$ nor ${\bf A}({\bf x})$.
This justification of complete positivity is obtained as
a part of an axiomatic characterization of physically realizable measurements \cite{Ozawa04,14MFQ-}.
Then Theorem \ref{CPIMP} enables us to regard the Davies-Lewis proposal restricted to CP instruments
as a statement that is consistent with the standard formulation of quantum mechanics 
and hence acceptable for physicists. The above discussion is summarized as follows.
\begin{dlocriterion}
For every apparatus ${\bf A}({\bf x})$ measuring ${\bf S}$,
where ${\bf x}$ is the output variable of ${\bf A}({\bf x})$
taking values in a measurable space $(S,\mathcal{F})$,
there always exists a CP instrument $\mathcal{I}$ for $(\mathcal{M},S)$
corresponding to ${\bf A}({\bf x})$ in the sense of the Davies-Lewis proposal,
i.e., for every input state $\rho$ and outcome $\Delta\in\mathcal{F}$ both
the probability distribution $\mathrm{Pr}\{{\bf x}\in\Delta\Vert\rho\}$ of ${\bf x}$
and the state $\rho_{\{{\bf x}\in\Delta\}}$ after the measurement are obtained from $\mathcal{I}$.
\end{dlocriterion}

\subsection{Kernels}
Here, we briefly summerize the theory of kernels.
We refer the reader to \cite{EvansLewis,Lance95,Skeide01} for standard references.

\begin{definition}[Kernel \text{\cite[p.11, ll.1--3]{EvansLewis}}]
Let $C$ be a set and $\mathcal{H}$ a Hilbert space.
A map $K:C\times C\rightarrow{\bf B}(\mathcal{H})$ is called a
kernel of $C$ on $\mathcal{H}$. $\mathbb{K}(C;\mathcal{H})$ denotes
the set of kernels of $C$ on $\mathcal{H}$.
\end{definition}
It should be noted that $\mathbb{K}(C;\mathcal{H})$ has
a natural ${\bf B}(\mathcal{H})$-bimodule structure.
\begin{definition}[\text{\cite[Definition 1.1]{EvansLewis}}]
Let $K\in\mathbb{K}(C;\mathcal{H})$. $K$ is said to be positive definite if
\begin{equation}
\sum_{i,j=1}^n \langle \xi_i| K(c_i,c_j)\xi_j\rangle\geq 0
\end{equation}
for every $n\in\mathbb{N}$, $c_1,c_2,\cdots,c_n\in C$ and $\xi_1,\xi_2,\cdots,\xi_n\in\mathcal{H}$.
$\mathbb{K}(C;\mathcal{H})_+$ denotes
the set of positive definite kernels of $C$ on $\mathcal{H}$.
\end{definition}

\begin{definition}[Kolmogorov decomposition \text{\cite[Definition 1.3]{EvansLewis}}]
Let $K\in\mathbb{K}(C;\mathcal{H})$. A pair $(\mathcal{K},\Lambda)$
of a Hilbert space $\mathcal{K}$ and a map $\Lambda:C\rightarrow {\bf B}(\mathcal{H},\mathcal{K})$
is called a Kolmogorov decomposition of $K$ if it satisfies
\begin{equation}
K(c,c^\prime)=\Lambda(c)^\ast \Lambda(c^\prime)
\end{equation}
for all $c,c^\prime\in C$. A Kolmogorov decomposition $(\mathcal{K},\Lambda)$ of $K$
is said to be minimal if $\mathcal{K}=\overline{\mathrm{span}}(\Lambda(C)\mathcal{H})$.
\end{definition}
The following representation theorem holds for kernels.
\begin{theorem}[\text{\cite[Lemma 1.4, Theorems 1.8 and 1.9]{EvansLewis}}] \label{minKD}
Let $C$ be a set and $\mathcal{H}$ a Hilbert space.
For every $K\in\mathbb{K}(C;\mathcal{H})$, $K$ admits a Kolmogorov decomposition if and only if
it is an element of $\mathbb{K}(C;\mathcal{H})_+$.
For every $K\in\mathbb{K}(C;\mathcal{H})_+$, there exists
a minimal Kolmogorov decomposition $(\mathcal{K},\Lambda)$ of $K$, which is unique up to unitary equivalence.
\end{theorem}
This theorem is a key to the proof of the main theorem of the paper.
The famous Stinespring representation theorem is regarded as a corollary of this theorem.

\begin{theorem}[Arveson \text{\cite[Theorem 1.3.1]{Arveson69},
\cite[Theorem 12.7]{Paulsen02}}]\label{CommLift}
Let $\mathcal{H}$, $\mathcal{K}$ be Hilbert spaces, $\mathcal{B}$ a unital C$^\ast$-subalgebra of
${\bf B}(\mathcal{K})$ and $V$ an element
of ${\bf B}(\mathcal{H},\mathcal{K})$ such that
$\mathcal{K}=\overline{\mathrm{span}}(\mathcal{B}V\mathcal{H})$.
For every $A\in(V^\ast \mathcal{B}V)^\prime$, there exists a unique $A_1\in\mathcal{B}^\prime$ such that
$VA=A_1V$. Furthermore, the map $\pi^\prime:A\in(V^\ast \mathcal{B}V)^\prime\ni A\mapsto
A_1\in\mathcal{B}^\prime\cap\{VV^\ast\}^\prime$ is
an ultraweakly continuous surjective $^\ast$-homomorphism.
\end{theorem}
The following theorem holds as a corollary of \cite[Part I, Chapter 4, Theorem 3]{Dixmier},
\cite[Chapter IV, Theorem 5.5]{T79}:
\begin{theorem} \label{vNhom2}
Let $\mathcal{H}_1$ and $\mathcal{H}_2$ be Hilbert spaces.
If $\pi$ is a normal representation of ${\bf B}(\mathcal{H}_1)$ on
$\mathcal{H}_2$, there exist a Hilbert space $\mathcal{K}$
and a unitary operator $U$ of $\mathcal{H}_1\otimes\mathcal{K}$ onto $\mathcal{H}_2$ such that
\begin{equation}
\pi(X)=U(X\otimes 1_\mathcal{K})U^\ast
\end{equation}
for all $X\in{\bf B}(\mathcal{H}_1)$.
\end{theorem}
This theorem is also a key to the proof of the main theorem of the paper.

\section{Systems of Measurement Correlations}
\label{sec:3}
In this section, we introduce the concept of system of measurement correlations,
which is a natural, multivariate version of instrument and is defined as a family of multilinear maps
satisfying ``positive-definiteness", ``$\sigma$-additivity" and other conditions.
This is an appropriate abstraction of measurement correlations
in the context of quantum stochastic processes \cite{AFL82}.
It is known that the representation theory of CP instruments
contributed to quantum measurement theory \cite{Ozawa83,Oz84,OO16}.
Hence we adopt a representation-theoretical approach to system of measurement correlations.
The ``positive-definiteness" of systems of measurement correlations
enables us to apply the (minimal) Kolmogorov decomposition to
them, so that provides them with representation-theoretical structures.
As a result, a representation theorem (Theorem \ref{Stinespring})
similar to that for CP instruments \cite[Proposition 4.2]{Oz84} will be shown to hold
for systems of measurement correlations defined on an arbitrary von Neumann algebra.
To precisely understand physics described by systems of measurement correlations
we need a generalization of the Heisenberg picture
which is introduced after the proof of Theorem \ref{Stinespring}
and is called the generalized Heisenberg picture. The introduction of this new picture
is motivated also by the present circumstances that the understanding of the (usual) Heisenberg picture
has not been deepened in contrast to the Schr\"{o}dinger picture.
It should be stressed that the circumstances are never restricted to quantum measurement theory.

We adopt the following notations.
\begin{notation}
Let $\mathcal{T}^{(1)}$ be a set.
We define a set $\mathcal{T}$ by $\mathcal{T}=\cup_{j=1}^\infty (\mathcal{T}^{(1)})^j$.\\
$(i)$ For each $T\in\mathcal{T}$, $|T|$ denotes the natural number $n$ such that $T\in(\mathcal{T}^{(1)})^n$.\\
$(ii)$ For each $T=(t_1,t_2,\cdots,t_{n-1},t_n)\in\mathcal{T}$, we define 
$T^\#\in\mathcal{T}$ by $T^\#=(t_n,t_{n-1},\cdots,t_2,t_1)$.\\
$(iii)$ For any $T=(t_{1,1},\cdots,t_{1,m}),T_2=(t_{2,1},\cdots,t_{2,n})$ $\in\mathcal{T}$,
the product $T_1\times T_2$ is defined by
\begin{equation}
T_1\times T_2=(t_{1,1},\cdots,t_{1,m},t_{2,1},\cdots,t_{2,n}).
\end{equation}
Since it holds that $T_1\times(T_2\times T_3)=(T_1\times T_2)\times T_3$,
$T_1\times(T_2\times T_3)$ is written as $T_1\times T_2\times T_3$.\\
$(iv)$ For any $n\in\mathbb{N}$ and $\overrightarrow{M}=(M_1,M_2,\cdots,M_{n-1},M_n)\in\mathcal{M}^n$,
we define $\overrightarrow{M}^\#\in \mathcal{M}^n$ by
\begin{equation}
\overrightarrow{M}^\#=(M_n^\ast,M_{n-1}^\ast,\cdots,M_2^\ast,M_1^\ast).
\end{equation}
$(v)$ For any $m,n\in\mathbb{N}$, $\overrightarrow{M}_1=(M_{1,1},\cdots,M_{1,m})\in\mathcal{M}^m$
and $\overrightarrow{M}_2=(M_{2,1},\cdots,M_{2,n})\in\mathcal{M}^n$,
the product $\overrightarrow{M}_1\times\overrightarrow{M}_2\in\mathcal{M}^{m+n}$ is defined by
\begin{equation}
\overrightarrow{M}_1\times\overrightarrow{M}_2=(M_{1,1},\cdots,M_{1,m},M_{2,1},\cdots,M_{2,n}).
\end{equation}
Since it holds that $\overrightarrow{M}_1\times(\overrightarrow{M}_2\times \overrightarrow{M}_3)=(\overrightarrow{M}_1\times \overrightarrow{M}_2)\times \overrightarrow{M}_3$,
$\overrightarrow{M}_1\times(\overrightarrow{M}_2\times \overrightarrow{M}_3)$ is written as $\overrightarrow{M}_1\times \overrightarrow{M}_2\times \overrightarrow{M}_3$.\\
\end{notation}

In addition, for every family $\{\Pi_t\}_{t\in\mathcal{T}^{(1)}}$ of representations
of $\mathcal{M}$ on a Hilbert space $\mathcal{L}$, we adopt the notation
\begin{equation}
\Pi_T(\overrightarrow{M})=\Pi_{t_1}(M_1)\cdots \Pi_{t_{|T|}}(M_{|T|})
\end{equation}
for all $T=(t_1,\cdots,t_{|T|})\in\mathcal{T}$ and
$\overrightarrow{M}=(M_1,\cdots,M_{|T|})\in\mathcal{M}^{|T|}$.

Let $(S,\mathcal{F})$ be a measurable space. We define a set $\mathcal{T}_S$ by
\begin{align}
\mathcal{T}_S &= \cup_{j=1}^\infty \;(\mathcal{T}_S^{(1)})^j,\\
\mathcal{T}_S^{(1)} &= \{in\}\cup \mathcal{F},
\end{align}
where $in$ is a symbol.

We shall define the notion of system of measurement correlations, which is a modified version of
projective system of multikernels analyzed in the previous investigations \cite{AFL82,Belavkin85}.
We define and analyze
only the case that systems of measurement correlations do not have explicit time-dependence
for simplicity herein.

\begin{definition}[System of measurement correlations]
A family $\{W_T\}_{T\in\mathcal{T}}$ of maps
$W_T:\mathcal{M}^{|T|}=\overbrace{\mathcal{M}\times\cdots\times\mathcal{M}}^{|T|} \rightarrow\mathcal{M}$
is called a system of measurement correlations
for $(\mathcal{M},S)$ if it satisfies $\mathcal{T}^{(1)}=\mathcal{T}^{(1)}_S$
and the following six conditions:\\
$(\mathrm{MC}1)$ For any $T\in\mathcal{T}$, $W_T(M_1,\cdots,M_{|T|})$ is
separately linear and ultraweakly continuous in each variable $M_1,\cdots,M_{|T|}\in\mathcal{M}$.\\
$(\mathrm{MC}2)$ For any $n\in\mathbb{N}$, $(T_1,\overrightarrow{M}_1),\cdots,(T_n,\overrightarrow{M}_n)$
$\in\cup_{T\in\mathcal{T}}(\{T\}\times\mathcal{M}^{|T|})$,
and $\xi_1,\cdots,\xi_n\in\mathcal{H}$,
\begin{equation}
\sum_{i,j=1}^n
\langle\xi_i|W_{T_i^\#\times T_j}(\overrightarrow{M}_i^\#\times
\overrightarrow{M}_j)\xi_j\rangle \geq 0.
\end{equation}
$(\mathrm{MC}3)$ For any $T=(t_1,\cdots,t_{|T|})\in\mathcal{T}$,
$\overrightarrow{M}=(M_1,\cdots,M_{|T|})\in \mathcal{M}^{|T|}$ and $M\in\mathcal{M}$,
\begin{align}
MW_T(\overrightarrow{M}) &=W_{(in)\times T}((M)\times \overrightarrow{M}),\\
W_T(\overrightarrow{M})M &=W_{T\times (in)}(\overrightarrow{M}\times (M)).
\end{align}
$(\mathrm{MC}4)$ Let $T=(t_1,\cdots,t_{|T|})\in\mathcal{T}$.
If $t_k=t_{k+1}=in$ or $t_k,t_{k+1}\in\mathcal{F}$ for some $1\leq k\leq |T|-1$,
\begin{equation}
W_T(M_1,\cdots,M_k,M_{k+1},\cdots,M_{|T|})= W_{T^\prime}(M_1,\cdots,M_kM_{k+1},\cdots,M_{|T|})
\end{equation}
for all $(M_1,\cdots,M_{|T|})\in \mathcal{M}^{|T|}$,
where $T^\prime=(t_1,\cdots,t_{k-1},t_k\cap t_{k+1},t_{k+2}\cdots,t_{|T|})$ and
\begin{equation*}
t_k\cap t_{k+1}=\left\{
\begin{array}{ll}
in, &\quad  (\mathrm{if}\; t_k=t_{k+1}=in)\\
 t_k\cap t_{k+1}, &\quad (\mathrm{if}\;t_k, t_{k+1}\in\mathcal{F}).
\end{array}
\right.
\end{equation*}
$(\mathrm{MC}5)$ For any $T=(t_1,\cdots,t_{|T|})\in\mathcal{T}$ with
$t_k=in$ or $S$, and $(M_1,\cdots,M_{|T|})\in \mathcal{M}^{|T|}$ with $M_k=1$,
\begin{equation}
W_T(M_1,\cdots,M_{k-1},1,M_{k+1},\cdots,M_{|T|})=W_{\hat{k}T}(M_1,\cdots,M_{k-1},M_{k+1},\cdots,M_{|T|}),
\end{equation}
where $\hat{k}T=(t_1,\cdots,t_{k-1},t_{k+1},t_{k+2}\cdots,t_{|T|})$. In addition,
\begin{equation}
W_{in}(1)=W_S(1)=1.
\end{equation}
$(\mathrm{MC}6)$ For any $n\in\mathbb{N}$, $1\leq k\leq n$,
$t_1,\cdots,t_{k-1},t_{k+1},\cdots,t_n\in\mathcal{T}^{(1)}$,
mutually disjoint sequence $\{t_{k,j}\}_j\subset\mathcal{F}$,
$\overrightarrow{M}\in\mathcal{M}^n$ and $\rho\in\mathcal{M}_{\ast}$,
\begin{equation}
\langle\rho,W_{(t_1,\cdots,t_{k-1},\cup_jt_{k,j},t_{k+1},\cdots,t_n)}(\overrightarrow{M})\rangle
=\sum_j \langle\rho,W_{(t_1,\cdots,t_{k-1},t_{k,j},t_{k+1},\cdots,t_n)}(\overrightarrow{M})\rangle.
\end{equation}
\end{definition}
$in$ and $\Delta\in\mathcal{F}$ are subscripts that specify
the time before the measurement and the time after the measurement, respectively.
In $W_{T}(\overrightarrow{M})$, components of $\overrightarrow{M}$ indexed by $in$ 
and those of $\overrightarrow{M}$ indexed by $\Delta\in\mathcal{F}$,
describe observables before the measurement and those after the measurement, respectively,
for each $T\in\mathcal{T}$ and $\overrightarrow{M}\in\mathcal{M}^{|T|}$.
Especially, the latter represents observables of the system after the measurement
in the situation that values of the output variable of the measuring appratus
are restricted to $\Delta\in\mathcal{F}$.
The discussion in Section \ref{sec:4} will support this interpretation.

It is easy to generalize systems of measurement correlations to the case that
they have explicit time-dependence by modifying the definition.
For this purpose, $\mathcal{T}_S^{(1)}$ is replaced by
$\mathcal{T}_{G,S}^{(1)}=\{in\}\cup (G\times\mathcal{F})$,
where $G$ is the set representing time and is usually assumed to be a subset of $\mathbb{R}$,
and, for instance, the condition $(\mathrm{MC}4)$ is replaced by\\
$(\mathrm{MC}4^\prime)$ Let $T=(t_1,\cdots,t_{|T|})\in\mathcal{T}$.
If $t_k=t_{k+1}=in$ or $t_k=(g,\Delta_k),t_{k+1}=(g,\Delta_{k+1})\in G\times\mathcal{F}$
for some $1\leq k\leq |T|-1$, then
\begin{equation}
W_T(M_1,\cdots,M_k,M_{k+1},\cdots,M_{|T|})= W_{T^\prime}(M_1,\cdots,M_kM_{k+1},\cdots,M_{|T|})
\end{equation}
for all $(M_1,\cdots,M_{|T|})\in \mathcal{M}^{|T|}$,
where $T^\prime=(t_1,\cdots,t_{k-1},t_k\cap t_{k+1},t_{k+2}\cdots,t_{|T|})$ and
\begin{equation*}
t_k\cap t_{k+1}=\left\{
\begin{array}{ll}
in, &\quad  (\mathrm{if}\; t_k=t_{k+1}=in)\\
 (g,\Delta_k\cap \Delta_{k+1}), &\quad
 (\mathrm{if}\;t_k=(g,\Delta_k),t_{k+1}=(g,\Delta_{k+1})\in G\times\mathcal{F}).
\end{array}
\right.
\end{equation*}
Other conditions are also modified in the same manner.

When a system $\{W_T\}_{T\in\mathcal{T}}$ of measurement correlations for $(\mathcal{M},S)$ is given,
an instrument $\mathcal{I}_W$ for $(\mathcal{M},S)$ is defined by
\begin{equation}
\mathcal{I}_W(M,\Delta)=W_\Delta(M)
\end{equation}
for all $\Delta\in\mathcal{F}$ and $M\in\mathcal{M}$,
which is seen to be completely positive by the condition (MC2).

Every system of measurement correlations admits the following representation theorem.
\begin{theorem}\label{Stinespring}
Let $\mathcal{M}$ be a von Neumann algebra
on a Hilbert space $\mathcal{H}$ and $(S,\mathcal{F})$ a measurable space.
For any systems $\{W_T\}_{T\in\mathcal{T}}$ of measurement correlations for $(\mathcal{M},S)$,
there exist a Hilbert space $\mathcal{L}$, a family $\{\Pi_t\}_{t\in\mathcal{T}^{(1)}}$
of normal $(^\ast$-$)$representations of $\mathcal{M}$ on $\mathcal{L}$
and an isometry $V$ from $\mathcal{H}$ to $\mathcal{L}$ such that
\begin{equation}
\Pi_{in}(M)V=VM
\end{equation}
for all $M\in\mathcal{M}$, and that
\begin{equation}
W_T(\overrightarrow{M})=V^\ast \Pi_{T}(\overrightarrow{M})V
\end{equation}
for all $T\in\mathcal{T}$ and $\overrightarrow{M}\in\mathcal{M}^{|T|}$.
\end{theorem}

\begin{proof}
Let $\{W_T\}_{T\in\mathcal{T}}$ be
a system of measurement correlations for $(\mathcal{M},S)$.
We set $\mathcal{C}=\cup_{T\in\mathcal{T}}(\{T\}\times\mathcal{M}^{|T|})$.
We define a kernel
$K:\mathcal{C}\times \mathcal{C}\rightarrow \mathcal{M}$ by
\begin{equation}
K({\bf a},{\bf b})=
W_{T_1^\#\times T_2}(\overrightarrow{M}_1^\#\times
\overrightarrow{M}_2)
\end{equation}
for all ${\bf a}=(T_1,\overrightarrow{M}_1)$,
${\bf b}=(T_2,\overrightarrow{M}_2)\in \mathcal{C}$.
By the definition of a system of measurement correlations,
$K$ is positive definite.
By Theorem \ref{minKD}, there exists 
the minimal Kolmogorov decomposition $(\mathcal{L},\Lambda)$ of $K$ such that
\begin{equation}
K({\bf a},{\bf b})=\Lambda({\bf a})^\ast \Lambda({\bf b})
\end{equation}
for all ${\bf a}=(T_1,\overrightarrow{M}_1)$,
${\bf b}=(T_2,\overrightarrow{M}_2)\in \mathcal{C}$. We remark that we use the fact that
$\mathrm{span}(\Lambda(\mathcal{C})\mathcal{H})$ is dense in $\mathcal{L}$ many times in this proof.

For each $t\in\mathcal{T}^{(1)}$ and $M\in\mathcal{M}$,
we define a map $\Pi_t(M)$ on $\mathrm{span}\;\Lambda(\mathcal{C})\mathcal{H}$ by
\begin{equation}
\Pi_t(M)\Lambda({\bf a})\xi=\Lambda({\bf a^\prime})\xi
\end{equation}
for all ${\bf a}=(T,\overrightarrow{M})=
((t_1,\cdots,t_{|T|}),(M_1,\cdots,M_{|T|}))\in \mathcal{C}$ and $\xi\in\mathcal{H}$,
where
\begin{equation}
{\bf a^\prime}=((t)\times T,(M)\times\overrightarrow{M})=((t,t_1,\cdots,t_{|T|}),(M,M_1,\cdots,M_{|T|})).
\end{equation}
For all $t\in\mathcal{T}^{(1)}$, we show that $\Pi_t:M\mapsto \Pi_t(M)$ is a normal $^\ast$-representation of
$\mathcal{M}$.
By the condition (MC1), it holds that
\begin{align}
 &\hspace{5mm} \langle \Lambda({\bf a})\xi_1|\Pi_t(\alpha M+\beta N) \Lambda({\bf b})\xi_2\rangle\nonumber\\
&=\langle \xi_1|\Lambda({\bf a})^\ast \Pi_t(\alpha M+\beta N) \Lambda({\bf b})\xi_2\rangle \nonumber\\
 &=\langle \xi_1|W_{T_1^\#\times(t)\times T_2}
 (\overrightarrow{M}_1^\#\times(\alpha M+\beta N)\times \overrightarrow{M}_2)\xi_2\rangle \nonumber \\
 &=\alpha\langle \xi_1|W_{T_1^\#\times(t)\times T_2}
 (\overrightarrow{M}_1^\#\times( M)\times \overrightarrow{M}_2)\xi_2\rangle \nonumber \\
 &\hspace{7mm}+\beta\langle \xi_1|W_{T_1^\#\times(t)\times T_2}
 (\overrightarrow{M}_1^\#\times(N)\times \overrightarrow{M}_2)\xi_2\rangle \nonumber \\
 &=\alpha\langle \xi_1|\Lambda({\bf a})^\ast \Pi_t(M) \Lambda({\bf b})\xi_2\rangle
 +\beta\langle \xi_1|\Lambda({\bf a})^\ast \Pi_t(N) \Lambda({\bf b})\xi_2\rangle \nonumber \\
 &=\langle \xi_1|\Lambda({\bf a})^\ast (\alpha\Pi_t(M)+\beta\Pi_t(N)) \Lambda({\bf b})\xi_2\rangle \nonumber \\
 &=\langle \Lambda({\bf a})\xi_1| (\alpha\Pi_t(M)+\beta\Pi_t(N)) \Lambda({\bf b})\xi_2\rangle 
\end{align}
for any $t\in\mathcal{T}^{(1)}$, $\alpha,\beta\in\mathbb{C}$, $M,N\in\mathcal{M}$,
${\bf a}=(T_1,\overrightarrow{M}_1)$, ${\bf b}=(T_2,\overrightarrow{M}_2)\in \mathcal{C}$ and $\xi_1,\xi_2\in\mathcal{H}$,
so that $\Pi_t(\alpha M+\beta N)=\alpha\Pi_t(M)+\beta\Pi_t(N)$
for all $t\in\mathcal{T}^{(1)}$, $\alpha,\beta\in\mathbb{C}$ and $M,N\in\mathcal{M}$.

Similarly, by the condition (MC4) it holds that
\begin{align}
 &\hspace{5mm} \langle \Lambda({\bf a})\xi_1|\Pi_t(M)\Pi_t(N) \Lambda({\bf b})\xi_2\rangle
 =\langle \xi_1|\Lambda({\bf a})^\ast\Pi_t(M)\Pi_t(N) \Lambda({\bf b})\xi_2\rangle \nonumber \\
 &=\langle \xi_1|W_{T_1^\#\times(t,t)\times T_2}
 (\overrightarrow{M}_1^\#\times(M, N)\times \overrightarrow{M}_2)\xi_2\rangle \nonumber \\
 &=\langle \xi_1|W_{T_1^\#\times(t)\times T_2}
 (\overrightarrow{M}_1^\#\times(MN)\times \overrightarrow{M}_2)\xi_2\rangle  \nonumber \\
 &=\langle \xi_1|\Lambda({\bf a})^\ast\Pi_t(MN) \Lambda({\bf b})\xi_2\rangle
 = \langle \Lambda({\bf a})\xi_1|\Pi_t(MN) \Lambda({\bf b})\xi_2\rangle 
\end{align}
for any $t\in\mathcal{T}^{(1)}$, $M,N\in\mathcal{M}$,
${\bf a}=(T_1,\overrightarrow{M}_1)$, ${\bf b}=(T_2,\overrightarrow{M}_2)\in \mathcal{C}$ and $\xi_1,\xi_2\in\mathcal{H}$,
so that $\Pi_t(MN)=\Pi_t(M)\Pi_t(N)$ for all $t\in\mathcal{T}^{(1)}$ and $M,N\in\mathcal{M}$.

For any $t\in\mathcal{T}^{(1)}$, $n\in\mathbb{N}$, ${\bf a}_1=(T_1,\overrightarrow{M}_1)$, ${\bf a}_2=(T_2,\overrightarrow{M}_2)$, $\cdots$,
${\bf a}_n=(T_n,\overrightarrow{M}_n)\in \mathcal{C}$ and $\xi_1,\xi_2,\cdots,\xi_n\in\mathcal{H}$,
the map 
\begin{equation}
\mathcal{M}\ni M\mapsto \sum_{i,j=1}^n \langle \Lambda({\bf a}_i)\xi_i|
\Pi_t(M)\Lambda({\bf a}_j)\xi_j \rangle\in\mathbb{C}
\end{equation}
is normal linear functional on $\mathcal{M}$, which is also positive since it holds
by the conditions (MC2) and (MC4) that
\begin{align}
&\hspace{5mm}\sum_{i,j=1}^n \langle \Lambda({\bf a}_i)\xi_i|\Pi_t(M)\Lambda({\bf a}_j)\xi_j \rangle
=\sum_{i,j=1}^n \langle \xi_i|W_{T_i^\#\times(t)\times T_j}
 (\overrightarrow{M}_i^\#\times(M)\times \overrightarrow{M}_j)\xi_j\rangle \nonumber\\
 &=\sum_{i,j=1}^n \langle \xi_i|W_{T_i^\#\times(t,t)\times T_j}
 (\overrightarrow{M}_i^\#\times(\sqrt{M},\sqrt{M})\times \overrightarrow{M}_j)\xi_j\rangle\nonumber \\
 &=\sum_{i,j=1}^n \langle \xi_i|W_{((t)\times T_i)^\#\times((t)\times T_j)}
 (((\sqrt{M})\times\overrightarrow{M}_i)^\#
 \times((\sqrt{M})\times \overrightarrow{M}_j))\xi_j\rangle \geq 0
\end{align}
for all $M\in\mathcal{M}_+=\{M\in\mathcal{M}\;|\;M\geq 0\}$, $t\in\mathcal{T}^{(1)}$,
$n\in\mathbb{N}$, ${\bf a}_1=(T_1,\overrightarrow{M}_1)$, ${\bf a}_2=(T_2,\overrightarrow{M}_2)$, $\cdots$,
${\bf a}_n=(T_n,\overrightarrow{M}_n)\in \mathcal{C}$ and $\xi_1,\xi_2,\cdots,\xi_n\in\mathcal{H}$.
Thus, for any $t\in\mathcal{T}^{(1)}$, $M\in\mathcal{M}$,
$n\in\mathbb{N}$, ${\bf a}_1=(T_1,\overrightarrow{M}_1)$, ${\bf a}_2=(T_2,\overrightarrow{M}_2)$, $\cdots$,
${\bf a}_n=(T_n,\overrightarrow{M}_n)\in \mathcal{C}$ and $\xi_1,\xi_2,\cdots,\xi_n\in\mathcal{H}$
we have
\begin{equation}
\left\Vert \sum_{i=1}^n \Pi_t(M)\Lambda({\bf a}_i)\xi_i\right\Vert
\leq \Vert M\Vert \cdot \left\Vert \sum_{i=1}^n \Lambda({\bf a}_i)\xi_i\right\Vert.
\end{equation}
For every $t\in\mathcal{T}^{(1)}$ and $M\in\mathcal{M}$, $\Pi_t(M)$ is a bounded operator on $\mathcal{L}$.
In addition, for all $t\in\mathcal{T}^{(1)}$, $M\in\mathcal{M}$ and
${\bf a}=(T_1,\overrightarrow{M}_1)$, ${\bf b}=(T_2,\overrightarrow{M}_2)\in \mathcal{C}$,
\begin{align}
 &\hspace{5mm}\langle \Lambda({\bf a})\xi_1|\Pi_t(M)^\ast \Lambda({\bf b})\xi_2\rangle 
 =\langle \xi_1| (\Pi_t(M)\Lambda({\bf a}))^\ast\Lambda({\bf b})\xi_2\rangle  \nonumber \\
 &=\langle \xi_1|W_{((t)\times T_1)^\#\times T_2}
 (((M)\times\overrightarrow{M}_1)^\#\times \overrightarrow{M}_2)\xi_2 \rangle \nonumber \\
 &=\langle \xi_1|W_{T_1^\#\times(t)\times T_2}
 (\overrightarrow{M}_1^\#\times(M^\ast)\times \overrightarrow{M}_2)\xi_2 \rangle \nonumber \\
 &=\langle \xi_1| \Lambda({\bf a})^\ast\Pi_t(M^\ast)\Lambda({\bf b})\xi_2\rangle
 =\langle \Lambda({\bf a})\xi_1|\Pi_t(M^\ast) \Lambda({\bf b})\xi_2\rangle.
\end{align}
Thus, for every $t\in\mathcal{T}^{(1)}$
$\Pi_t$ is a normal $^\ast$-representation of $\mathcal{M}$ on $\mathcal{L}$.
By the condition (MC5), $\Pi_{in}$ and $\Pi_S$ are nondegenerate,
i.e., $\Pi_{in}(1)=\Pi_S(1)=1_{{\bf B}(\mathcal{L})}$.

For every $t\in\mathcal{T}$ and and $\overrightarrow{M}\in\mathcal{M}^{|T|}$, it then holds that
\begin{align}
&\quad\;\;V^\ast \Pi_{T}(\overrightarrow{M})V =\Lambda((in,1))^\ast
\Pi_{T}(\overrightarrow{M})\Lambda((in,1)) \nonumber \\
 &=W_{(in)\times T\times (in)}((1)\times\overrightarrow{M}\times(1))=W_T(\overrightarrow{M}).
\end{align}
By the above relation and the condition (MC3),
we have $V^\ast \Pi_{in}(M)V=M$ for all $M\in\mathcal{M}$, and
\begin{align}
&\quad\vspace{5mm} (\Pi_{in}(M)V-VM)^\ast(\Pi_{in}(M)V-VM)\nonumber\\
&=V^\ast\Pi_{in}(M)^\ast\Pi_{in}(M)V-
V^\ast\Pi_{in}(M)^\ast VM-M^\ast V^\ast\Pi_{in}(M)V +M^\ast V^\ast VM \nonumber\\
 &=V^\ast\Pi_{in}(M^\ast M)V-
 V^\ast\Pi_{in}(M^\ast) VM-M^\ast V^\ast\Pi_{in}(M)V +M^\ast V^\ast VM \nonumber \\
 &= M^\ast M- M^\ast M-M^\ast M+M^\ast M=0
\end{align}
for all $M\in\mathcal{M}$, which implies $\Pi_{in}(M)V=VM$ for all $M\in\mathcal{M}$.
\end{proof}
\begin{remark}
By the above proof, we see the following.
For every family $\{W_T\}_{T\in\mathcal{T}}$ of maps
$W_T:\mathcal{M}^{|T|}\rightarrow\mathcal{M}$ satisfying the conditions
$(\mathrm{MC1})$, $(\mathrm{MC2})$, $(\mathrm{MC4})$, $(\mathrm{MC5})$ and $(\mathrm{MC6})$,
there exist a Hilbert space $\mathcal{L}$, a family $\{\Pi_t\}_{t\in\mathcal{T}^{(1)}}$
of normal $(^\ast$-$)$representations of $\mathcal{M}$ on $\mathcal{L}$
and an isometry $V$ from $\mathcal{H}$ to $\mathcal{L}$ such that
\begin{equation}\label{SRep}
W_T(\overrightarrow{M})=V^\ast \Pi_{T}(\overrightarrow{M})V
\end{equation}
for all $T\in\mathcal{T}$ and $\overrightarrow{M}\in\mathcal{M}^{|T|}$.
Eq.(\ref{SRep}) then implies
\begin{equation}
W_T(\overrightarrow{M})^\ast=W_{T^\#}(\overrightarrow{M}^\#)
\end{equation}
for all $T\in\mathcal{T}$ and $\overrightarrow{M}\in\mathcal{M}^{|T|}$.
\end{remark}

As seen the proof of Theorem \ref{Stinespring}, the following fact holds,
which will be used in the next section.
\begin{corollary}
Let $\mathcal{M}$ be a von Neumann algebra
on a Hilbert space $\mathcal{H}$ and $(S,\mathcal{F})$ a measurable space.
For any systems $\{W_T\}_{T\in\mathcal{T}}$ of measurement correlations for $(\mathcal{M},S)$,
let $(\mathcal{L}$, $\{\Pi_t\}_{t\in\mathcal{T}^{(1)}}$, $V)$ be a triplet in Theorem \ref{Stinespring}.
Then the map $\mathcal{F}\ni\Delta\mapsto
\Pi_\Delta(1)\in{\bf B}(\mathcal{L})$ is a projection-valued measure (PVM).
\end{corollary}
\begin{proof}
The proof can be easily done in terms of the conditions (MC4), (MC5) and (MC6).
\end{proof}

In \cite{AFL82}, a (noncommutative) stochastic process over
a C$^\ast$-algebra $\mathcal{B}$, indexed by a set $\mathbb{T}$,
is defined by a pair $(\mathcal{A},\{j_t\}_{t\in\mathbb{T}})$ of a C$^\ast$-algebra $\mathcal{A}$ and
a family $\{j_t\}_{t\in\mathbb{T}}$ of $^\ast$-homomorphisms from $\mathcal{B}$ into $\mathcal{A}$.
Obviously, a pair $({\bf B}(\mathcal{L}),\{\Pi_t\}_{t\in\mathcal{T}^{(1)}})$ in Theorem \ref{Stinespring} is
nothing but a stochastic process over a von Neumann algebra $\mathcal{M}$ indexed by $\mathcal{T}^{(1)}$
in this sense.

Let $\mathcal{M}$ be a von Neumann algebra on a Hilbert space.
Let $\mathbb{T}$ be a set. We set $\mathcal{T}_\mathbb{T}=\cup_{j=1}^\infty (\{in\}\cup\mathbb{T})^j$.
Let $(\mathcal{L},\{\Pi_t\}_{t\in\{in\}\cup\mathbb{T}},V)$ be a triplet consisting of a Hilbert space $\mathcal{L}$,
a family $\{\Pi_t\}_{t\in\{in\}\cup\mathbb{T}}$ of normal representations of $\mathcal{M}$ on $\mathcal{L}$
and $V$ an isometry from $\mathcal{H}$ to $\mathcal{L}$ such that
$\Pi_{in}(M)V=VM$ for all $M\in\mathcal{M}$ and that $V^\ast \Pi_T(\overrightarrow{M})V\in\mathcal{M}$
for all $T\in \mathcal{T}_\mathbb{T}$ and $\overrightarrow{M}\in \mathcal{M}^{|T|}$.
The generalized Heisenberg picture
is then formulated by this triple $(\mathcal{L},\{\Pi_t\}_{t\in\{in\}\cup\mathbb{T}},V)$,
which enables us to compare the situtation before the change, specified by a representation $\Pi_{in}$,
with the situtation after the change, specified by $\{\Pi_t\}_{t\in\mathbb{T}}$.
This interpretation naturally follows from the intertwining relation $\Pi_{in}(M)V=VM$ for all $M\in\mathcal{M}$
and from the generation of correlation functions $\mathcal{W}_T(\overrightarrow{M})=V^\ast \Pi_T(\overrightarrow{M})V$
for all $T\in \mathcal{T}_\mathbb{T}$ and $\overrightarrow{M}\in \mathcal{M}^{|T|}$.
For example, in a triplet $(\mathcal{L},\{\Pi_t\}_{t\in\mathcal{T}^{(1)}},V)$ in Theorem \ref{Stinespring},
$\Pi_{in}$ and $\{\Pi_t\}_{t\in\mathcal{F}}$ correspond to a representation before the measurement
and those after the measurement, respectively.
The author believes that the generalized Heisenberg picture introduced here
gives a right extension of the description of dynamical processes in the standard formulation of
quantum mechanics since it succeeds to the advantage of the (usual) Heisenberg picture
that we can calculate correlation functions of observables at different times.
This topic will be discussed in detail in the succeeding paper of the author.

\section{Unitary Dilation Theorem}
\label{sec:4}
As previously mentioned, the introduction of the concept of measuring process
was cruicial for the progress of the theory of quantum measurement and of instruments.
Measuring processes redefined as follows also play
the central role in quantum measurement theory based on the generalized Heisenberg picture.
\begin{definition}\label{MP2}
A measuring process $\mathbb{M}$ for $(\mathcal{M},S)$
is a 4-tuple $\mathbb{M}=(\mathcal{K},\sigma,E,U)$ which
consists of a Hilbert space $\mathcal{K}$, a normal state $\sigma$ on ${\bf B}(\mathcal{K})$,
a spectral measure $E:\mathcal{F} \rightarrow {\bf B}(\mathcal{K})$,
and a unitary operator $U$ on $\mathcal{H}\otimes\mathcal{K}$
and defines a system of measurement correlations $\{W_T^\mathbb{M}\}_{T\in\mathcal{T}_S}$ for
$(\mathcal{M},S)$ as follows: We define a representation $\pi_{in}$ of $\mathcal{M}$ and 
a family $\{\pi_\Delta\}_{\Delta\in\mathcal{F}}$ of those
of $\mathcal{M}$ on $\mathcal{H}\otimes\mathcal{K}$ by 
\begin{equation}
\pi_{in}(M) =M\otimes 1_\mathcal{K},\hspace{3mm}
\pi_\Delta(M) =U^\ast (M\otimes E(\Delta))U
\end{equation}
for all $M\in \mathcal{M}$ and $\Delta\in\mathcal{F}$, respectively. We use the notation
\begin{equation}
\pi_T(\overrightarrow{M})=\pi_{t_1}(M_1)\cdots \pi_{t_{|T|}}(M_{|T|})
\end{equation}
for all $T=(t_1,\cdots,t_{|T|})\in\mathcal{T}$ and
$\overrightarrow{M}=(M_1,\cdots,M_{|T|})\in\mathcal{M}^{|T|}$.
For each $T\in\mathcal{T}_S$, $W_T^\mathbb{M}:\mathcal{M}^{|T|}\rightarrow\mathcal{M}$ is defined by
\begin{equation}
W_T^\mathbb{M}(\overrightarrow{M})=(\mathrm{id}\otimes\sigma)(\pi_T(\overrightarrow{M}))
\end{equation}
for all $\overrightarrow{M}\in\mathcal{M}^{|T|}$.
\end{definition}
It is easily seen that two definitions of measuring processes
for $({\bf B}(\mathcal{H}),S)$ are equivalent.

We say that a CP instrument $\mathcal{I}$ for $(\mathcal{M},S)$
is \textit{realized} by a measuring process $\mathbb{M}$ for $(\mathcal{M},S)$
in the sense of Definition \ref{MP2},
or $\mathbb{M}$ \textit{realizes} $\mathcal{I}$ if $\mathcal{I}=\mathcal{I}_\mathbb{M}$.
$\mathrm{CPInst}_{\mathrm{RE}}(\mathcal{M},S)$ denotes the set of
CP instruments for $(\mathcal{M},S)$ realized by measuring processes for $(\mathcal{M},S)$
in the sense of Definition \ref{MP2}.
Then we have $\mathrm{CPInst}_{\mathrm{RE}}(\mathcal{M},S)
\subseteq\mathrm{CPInst}_{\mathrm{NE}}(\mathcal{M},S)$.
It will be shown in Section \ref{sec:6} that
\begin{equation}
\mathrm{CPInst}_{\mathrm{RE}}(\mathcal{M},S)=\mathrm{CPInst}_{\mathrm{NE}}(\mathcal{M},S).
\end{equation}

\begin{definition}
Let $n\in\mathbb{N}$.
Two measuring processes $\mathbb{M}_1$ and $\mathbb{M}_2$ for $(\mathcal{M},S)$
are said to be $n$-equivalent if
$W_T^{\mathbb{M}_1}=W_T^{\mathbb{M}_2}$ for all $T\in\mathcal{T}$ such that $|T|\leq n$.
Two measuring processes $\mathbb{M}_1$
and $\mathbb{M}_2$ for $(\mathcal{M},S)$
are said to be completely equivalent if
they are $n$-equivalent for all $n\in\mathbb{N}$.
\end{definition}
The $n$-equivalence class of a measuring process $\mathbb{M}$ for $(\mathcal{M},S)$
is nothing but the set of measuring processes $\mathbb{M}^\prime$ for $(\mathcal{M},S)$ whose
correlation functions of order less or equal to $n$ are identical to those defined by $\mathbb{M}$, i.e.,
$W_T^\mathbb{M}=W_T^{\mathbb{M}^\prime}$ for all $T\in\mathcal{T}$ such that $|T|\leq n$.
Since a measuring process $\mathbb{M}$ for $(\mathcal{M},S)$ in the sense of Definition \ref{MP2}
is also that in the sense of Definition \ref{MP1}, the statistical equivalence works for the former.
Of course, the $2$-equivalence is the same as the statistical equivalence.
In practical situations, dynamical aspects of physical systems
are usually analyzed in terms of correlation functions of finite order.
Thus it is natural to consider that the classification of measuring processes
by the $n$-equivalence for not so large $n\in\mathbb{N}$ is valid in the same way.
It should be stressed here that causal relations cannot be verified without using
correlation functions (of observables at different times)
and that situations concerned with measurements are not the exception.
A successful example of causal relations in the context of measurement
has been already given by the notion of perfect correlation introduced in \cite{Ozawa06},
which uses correlation functions of order $2$.
One may consider that the complete equivalence of measuring processes is unrealistic and useless,
but we believe that it is much useful since the following theorem holds.
\begin{theorem}\label{OTOC}
Let $\mathcal{H}$ be a Hilbert space and $(S,\mathcal{F})$ a measurable space.
Then there is a one-to-one correpondence
between complete equivalence classes of measuring processes $\mathbb{M}=(\mathcal{K},\sigma,E,U)$
for $({\bf B}(\mathcal{H}),S)$
and systems $\{W_T\}_{T\in\mathcal{T}}$ of measurement correlations for $({\bf B}(\mathcal{H}),S)$,
which is given by the relation
\begin{equation}\label{SEMPandCP}
W_T(\overrightarrow{M})=W_T^\mathbb{M}(\overrightarrow{M})
\end{equation}
for all $T\in\mathcal{T}$ and $\overrightarrow{M}\in\mathcal{M}^{|T|}$.
\end{theorem}
Let $\mathcal{H}_1$ and $\mathcal{H}_2$ be Hilbert spaces. For each $\eta\in\mathcal{H}_2$,
we define a linear map $V_\eta:\mathcal{H}_1\rightarrow \mathcal{H}_1\otimes\mathcal{H}_2$
by $V_\eta\xi=\xi\otimes\eta$ for all $\xi\in\mathcal{H}_1$.
It is easily seen that, for each $\eta\in\mathcal{H}_2$, $V_\eta$ satisfies
$(X\otimes 1)V_\eta=V_\eta X$ for all $X\in {\bf B}(\mathcal{H}_1)$.
For any $x\in\mathcal{H}_1\backslash \{0\}$, $P_x$ denotes
the projection from $\mathcal{H}_1$ onto the linear subspace $\mathbb{C}x$
of $\mathcal{H}_1$ linearly spanned by $x$. For any $x,y\in\mathcal{H}_1$,
we define $|y\rangle\langle x| \in {\bf B}(\mathcal{H}_1)$ by
$|y\rangle\langle x|z=\langle x|z \rangle y$ for all $z\in\mathcal{H}_1$.

\begin{lemma}\label{intertwine}
Let $\mathcal{H}_1$ and $\mathcal{H}_2$ be Hilbert spaces.
Let $V$ be an isometry from $\mathcal{H}_1$ to $\mathcal{H}_1\otimes\mathcal{H}_2$.
If $V$ satisfies $(X\otimes 1)V=VX$ for all $X\in {\bf B}(\mathcal{H}_1)$,
then there exists $\eta\in\mathcal{H}_2$ such that $V=V_\eta$.
\end{lemma}

\begin{proof}
Let $x\in\mathcal{H}_1\backslash \{0\}$.
Since $(P_x\otimes 1)Vx= VP_x x=Vx$, it holds that $Vx\in \mathbb{C}x\otimes\mathcal{H}_2$.
Hence, for any $x\in\mathcal{H}_1\backslash \{0\}$,
there is $\eta_x\in\mathcal{H}_2$ such that $Vx=x\otimes\eta_x$ and that $\Vert\eta_x\Vert=1$.

For any $x,y\in\mathcal{H}_1\backslash \{0\}$,
\begin{align}
 &\quad \langle x|x \rangle y\otimes \eta_y= \langle x|x \rangle Vy= V(\langle x|x \rangle y)
= V(|y\rangle\langle x| x) \nonumber\\
 &= (|y\rangle\langle x|\otimes 1)Vx=(|y\rangle\langle x|\otimes 1)(x\otimes \eta_x)
=\langle x|x \rangle y\otimes \eta_x
\end{align}
Thus $\eta_x=\eta_y$ for all $x,y\in\mathcal{H}_1\backslash \{0\}$.
This fact implies that the range of the map $\mathcal{H}_1\backslash \{0\}\ni x\mapsto\eta_x\in\mathcal{H}_2$
is one point.
We put $\eta=\eta_x$ for some $x\in\mathcal{H}_1\backslash \{0\}$.
By the linearity of $V$, $V0=0=0\otimes \eta$.
Thus we have $V=V_\eta$.
\end{proof}

\begin{proof}[Proof of Theorem \ref{OTOC}]
Let $\{W_T\}_{T\in\mathcal{T}}$ be
a system of measurement correlations for $({\bf B}(\mathcal{H}),S)$.
By Theorem \ref{Stinespring},
there exist a Hilbert space $\mathcal{L}_0$, a family $\{\Pi_t\}_{t\in\mathcal{T}^{(1)}}$
of normal representations of ${\bf B}(\mathcal{H})$ on $\mathcal{L}_0$
and an isometry $V_0$ from $\mathcal{H}$ to $\mathcal{L}_0$ such that
\begin{equation}
W_T(\overrightarrow{M})=V_0^\ast \Pi_{T}(\overrightarrow{M})V_0
\end{equation}
for all $T\in\mathcal{T}$ and $\overrightarrow{M}\in{\bf B}(\mathcal{H})^{|T|}$.

By Theorem \ref{vNhom2}, there exist a Hilbert space $\mathcal{L}_1$ and
a unitary operator $U_1:\mathcal{L}_0\rightarrow \mathcal{H}\otimes\mathcal{L}_1$ such that
\begin{equation}
\Pi_{in}(M)=U_1^\ast(M\otimes 1)U_1
\end{equation}
for all $M\in{\bf B}(\mathcal{H})$.
Similarly, by Theorem \ref{vNhom2}, there exist a Hilbert space $\mathcal{L}_2$ and
a unitary operator $U_2:\mathcal{L}_0\rightarrow \mathcal{H}\otimes\mathcal{L}_2$ such that
\begin{equation}
\Pi_S(M)=U_2^\ast(M\otimes 1)U_2
\end{equation}
for all $M\in{\bf B}(\mathcal{H})$,
and by Theorem \ref{CommLift} there exist a PVM $E_0:\mathcal{F}\rightarrow{\bf B}(\mathcal{L}_2)$
such that
\begin{equation}
\Pi_\Delta(1)=U_2^\ast(1\otimes E_0(\Delta))U_2
\end{equation}
for all $\Delta\in\mathcal{F}$.

We define a linear map $V:\mathcal{H}\rightarrow \mathcal{H}\otimes\mathcal{L}_1$
by $V=U_1V_0$, which is obviously seen to be an isometry.
Here, it holds that $V^\ast(M\otimes 1)V=M$ for all $M\in{\bf B}(\mathcal{H})$.
For all $M\in{\bf B}(\mathcal{H})$,
\begin{align}
&\;\hspace{5mm}((M\otimes 1)V-VM)^\ast ((M\otimes 1)V-VM) \nonumber\\
 &=V^\ast(M^\ast M\otimes 1)V-V^\ast(M^\ast \otimes 1)VM-M^\ast V^\ast(M \otimes 1)V+M^\ast V^\ast VM \nonumber\\
 &=M^\ast M-M^\ast\cdot M-M^\ast\cdot M+M^\ast M=0.
\end{align}
Thus we have $(M\otimes 1)V=VM$ for all $M\in{\bf B}(\mathcal{H})$.
By Lemma \ref{intertwine}, there is $\eta_1\in\mathcal{L}_1$ such that $V=V_{\eta_1}$.

Let $\eta_2\in\mathcal{L}_2$ such that $\Vert\eta_2\Vert=1$.
Let $\zeta$ be an isomorphism from $\mathcal{L}_1\otimes\mathcal{L}_2$ to
$\mathcal{L}_2\otimes\mathcal{L}_1$ defined by $\zeta(\xi_1\otimes\xi_2)=\xi_2\otimes\xi_1$
for all $\xi_1\in\mathcal{L}_1$ and $\xi_2\in\mathcal{L}_2$.
We define a unitary operator $U_3$ from $\mathcal{H}\otimes\mathcal{L}_1\otimes\mathbb{C}\eta_2$
to $\mathcal{H}\otimes\mathbb{C}\eta_1\otimes\mathcal{L}_2$ by
\begin{equation}
U_3(\xi\otimes\eta_2)=(1\otimes \zeta)(U_2U_1^\ast\xi\otimes\eta_1)
\end{equation}
for all $\xi\in\mathcal{H}\otimes\mathcal{L}_1$.
We define a unitary operator $U_5$ from $\mathbb{C}\eta_2$ to $\mathbb{C}\eta_1$
by $U_5x=\langle \eta_2|x \rangle \eta_1$ for all $x\in\mathbb{C}\eta_2$. Then $U_3$ has the following form:
\begin{equation}
U_3=(1\otimes \zeta)(U_2U_1^\ast\otimes U_5).
\end{equation}
Since both $\mathcal{H}\otimes\mathcal{L}_1\otimes\mathbb{C}\eta_2$
and $\mathcal{H}\otimes\mathbb{C}\eta_1\otimes\mathcal{L}_2$
are subspaces of $\mathcal{H}\otimes\mathcal{L}_1\otimes\mathcal{L}_2$ and satisfies
$\dim(\mathcal{H}\otimes\mathcal{L}_1\otimes\mathbb{C}\eta_2)=
\dim(\mathcal{H}\otimes\mathbb{C}\eta_1\otimes\mathcal{L}_2)$
by the above observation, it holds that
\begin{equation}
\dim((\mathcal{H}\otimes\mathcal{L}_1\otimes\mathbb{C}\eta_2)^\bot)=
\dim((\mathcal{H}\otimes\mathbb{C}\eta_1\otimes\mathcal{L}_2)^\bot).
\end{equation}
This fact implies that
there is a unitary operator $U_4$ from $(\mathcal{H}\otimes\mathcal{L}_1\otimes\mathbb{C}\eta_2)^\bot$
to $(\mathcal{H}\otimes\mathbb{C}\eta_1\otimes\mathcal{L}_2)^\bot$.
Let $Q$ be a projection operator from $\mathcal{H}\otimes\mathcal{L}_1\otimes\mathcal{L}_2$
onto $\mathcal{H}\otimes\mathcal{L}_1\otimes\mathbb{C}\eta_2$, i.e., $Q=1\otimes 1 \otimes P_{\eta_2}$.
Let $R$ be a projection operator from $\mathcal{H}\otimes\mathcal{L}_1\otimes\mathcal{L}_2$
onto $\mathcal{H}\otimes\mathbb{C}\eta_1\otimes\mathcal{L}_2$,
i.e., $R=1\otimes P_{\eta_1}\otimes 1$.
We then define a unitary operator $U$ on $\mathcal{H}$ by $U=U_3Q+U_4(1-Q)$.
It is obvious that $U$ satisfies $UQ=U_3Q=RU_3=RU$.

We define a Hilbert space $\mathcal{K}$ by
$\mathcal{K}=\mathcal{L}_1\otimes\mathcal{L}_2$, a normal state $\sigma$ on ${\bf B}(\mathcal{K})$ by
\begin{equation}
\sigma(Y)=\langle \eta_1\otimes\eta_2|Y(\eta_1\otimes\eta_2) \rangle
\end{equation}
for all $Y\in {\bf B}(\mathcal{K})$, and a spectral measure $E:\mathcal{F}\rightarrow
{\bf B}(\mathcal{K})$ by
\begin{equation}
E(\Delta)=1\otimes E_0(\Delta)
\end{equation}
for all $\Delta\in\mathcal{F}$. 

We show that the 4-tuple $\mathbb{M}:=(\mathcal{K},\sigma,E,U)$ is a measuring process
for $({\bf B}(\mathcal{H}),S)$ such that
\begin{equation}
W_T(\overrightarrow{M})=W_T^\mathbb{M}(\overrightarrow{M})
\end{equation}
for all $T\in\mathcal{T}$ and $\overrightarrow{M}\in{\bf B}(\mathcal{H})^{|T|}$.
Since $Q=1\otimes 1 \otimes P_{\eta_2}$ and
\begin{equation}
\pi_{in}(M)=U_1\Pi_{in}(M)U_1^\ast \otimes 1_{{\bf B}(\mathcal{L}_2)}
\end{equation}
for all $M\in{\bf B}(\mathcal{H})$, we have
\begin{equation} \label{In}
\pi_{in}(M)Q=Q\pi_{in}(M)
\end{equation}
for all $M\in{\bf B}(\mathcal{H})$. Similarly, we have
\begin{align}
\pi_\Delta(M)Q &=U^\ast(M\otimes E(\Delta))UQ \nonumber\\
 &=U^\ast(M\otimes E(\Delta))RU_3Q \nonumber \\
 &=U^\ast R(M\otimes E(\Delta))U_3Q  \nonumber \\
 &= U_3^\ast R(M\otimes E(\Delta))U_3Q \nonumber \\
 &= U_3^\ast (M\otimes E(\Delta))U_3Q \nonumber \\
 &= ((1\otimes \zeta)(U_2U_1^\ast\otimes U_5))^\ast (M\otimes E(\Delta))((1\otimes \zeta)(U_2U_1^\ast\otimes U_5))
 \nonumber \\
 &= (U_1U_2^\ast\otimes U_5^\ast)(M\otimes \zeta^\ast E(\Delta)\zeta)(U_2U_1^\ast\otimes U_5)Q  \nonumber \\
 &= (U_1U_2^\ast\otimes U_5^\ast)(M\otimes E_0(\Delta)\otimes 1_{{\bf B}(\mathcal{L}_1)})
 (U_2U_1^\ast\otimes U_5))Q  \nonumber \\
 &= (U_1\Pi_\Delta(M)U_1^\ast\otimes P_{\eta_2})Q\nonumber \\
 &= (U_1\Pi_\Delta(M)U_1^\ast\otimes 1_{{\bf B}(\mathcal{L}_2)})Q,
\end{align}
and
\begin{equation}\label{Out}
\pi_\Delta(M)Q=Q\pi_\Delta(M)
\end{equation}
for all $M\in{\bf B}(\mathcal{H})$ and $\Delta\in\mathcal{F}$.
By Eqs. (\ref{In}) and (\ref{Out}), it holds that
\begin{equation}
Q\pi_T(\overrightarrow{M})Q=
Q(U_1\Pi_{t_1}(M_1)\cdots \Pi_{t_{|T|}}(M_{|T|})U_1^\ast
\otimes 1_{{\bf B}(\mathcal{L}_2)})Q
\end{equation}
for all $T=(t_1,\cdots,t_{|T|})\in\mathcal{T}$ and
$\overrightarrow{M}=(M_1,\cdots,M_{|T|})\in{\bf B}(\mathcal{H})^{|T|}$.

For all $\xi\in\mathcal{H}$, $T=(t_1,\cdots,t_{|T|})\in\mathcal{T}$ and
$\overrightarrow{M}=(M_1,\cdots,M_{|T|})\in{\bf B}(\mathcal{H})^{|T|}$.
\begin{align}
\langle \xi| W_T^\mathbb{M}(\overrightarrow{M})\xi\rangle
&=\langle \xi| (\mathrm{id}\otimes\sigma)(\pi_T(\overrightarrow{M}))\xi\rangle \nonumber\\
 &= \langle \xi\otimes\eta_1\otimes\eta_2|\pi_T(\overrightarrow{M})
 (\xi\otimes\eta_1\otimes\eta_2)\rangle\nonumber \\
 &= \langle Q(\xi\otimes\eta_1\otimes\eta_2)|\pi_T(\overrightarrow{M})
 Q(\xi\otimes\eta_1\otimes\eta_2)\rangle\nonumber \\
 &= \langle V\xi\otimes\eta_2|Q\pi_T(\overrightarrow{M})
 Q(V\xi\otimes\eta_2)\rangle\nonumber \\
 &= \langle V\xi\otimes\eta_2|Q(U_1\Pi_{t_1}(M_1)\cdots \Pi_{t_{|T|}}(M_{|T|})U_1^\ast
\otimes 1_{{\bf B}(\mathcal{L}_2)})Q(V\xi\otimes\eta_2)\rangle\nonumber \\
 &= \langle V\xi\otimes\eta_2|(U_1\Pi_{t_1}(M_1)\cdots \Pi_{t_{|T|}}(M_{|T|})U_1^\ast
\otimes 1_{{\bf B}(\mathcal{L}_2)})(V\xi\otimes\eta_2)\rangle\nonumber \\
 &= \langle V\xi|U_1\Pi_{t_1}(M_1)\cdots \Pi_{t_{|T|}}(M_{|T|})U_1^\ast
 V\xi\rangle\nonumber \\
 &= \langle \xi|V^\ast U_1\Pi_{t_1}(M_1)\cdots \Pi_{t_{|T|}}(M_{|T|})U_1^\ast
 V\xi\rangle\nonumber \\
 &= \langle \xi|V_0^\ast \Pi_{t_1}(M_1)\cdots \Pi_{t_{|T|}}(M_{|T|})V_0
 \xi\rangle\nonumber \\
 &= \langle \xi| W_T(\overrightarrow{M})\xi\rangle,
\end{align}
which completes the proof.
\end{proof}
\begin{remark}
We adopt here the same notations as in the proof of the above theorem.
Suppose that $\mathcal{H}$ is separable and $(S,\mathcal{F})$ is a standard Borel space.
Let $\{\Delta_n\}_{n\in\mathbb{N}}$ be a countable generator of $\mathcal{F}$,
$\{M_n\}_{n\in\mathbb{N}}$ a dense subset of ${\bf B}(\mathcal{H})$ in the strong topology,
and $\{\xi_n\}_{n\in\mathbb{N}}$ a dense subset of $\mathcal{H}$.
Let $\{C_n\}_{n\in\mathbb{N}}$ be a well-ordering of
the countable set $\{\Pi_{in}(M_n)\;|\;n\in\mathbb{N}\}\cup\{\Pi_{\Delta_m}(M_n)\;|\;m,n\in\mathbb{N}\}$.
$\mathcal{L}_0$ has the increasing sequence $\{\mathcal{L}_{0,n}\}_{n\in\mathbb{N}}$
of separable closed subspaces, defined by
\begin{equation}
\mathcal{L}_{0,n}=\overline{\mathrm{span}}\{C_{f(1)}\cdots C_{f(n)}V_0\xi_k
\;|\;f\in\mathbb{N}^{\{1,\cdots,n\}},k\in\mathbb{N}\}
\end{equation}
for all $n\in\mathbb{N}$, such that $\mathcal{L}_0=\overline{\mathrm{span}}(\cup_{n}\mathcal{L}_{0,n})$,
where $\mathbb{N}^{\{1,\cdots,n\}}$ is the set of maps from
$\{1,\cdots,n\}$ to $\mathbb{N}$.
Hence, $\mathcal{L}_0$ is separable because we have
\begin{equation}
\mathcal{L}_0= \overline{\mathrm{span}}\left\{
\bigoplus_{n=1}^\infty \;(\mathcal{L}_{0,n-1})^\bot\cap \mathcal{L}_{0,n}\right\},
\end{equation}
where $\mathcal{L}_{0,0}=\{0\}$.
It is immediately seen that $\mathcal{L}_1$, $\mathcal{L}_2$
and $\mathcal{K}=\mathcal{L}_1\otimes\mathcal{L}_2$ are also separable.
\end{remark}


\section{Extendability of CP instruments to systems of measurement correlations}
\label{sec:5}
To begin with, the following theorem similar to \cite[Theorem 3.4]{OO16}
holds for arbitrary von Neumann algebras $\mathcal{M}$.
\begin{corollary}
Let $\mathcal{M}$ be a von Neumann algebra on a Hilbert space $\mathcal{H}$
and $(S,\mathcal{F})$ a measurable space.
For a system $\{W_T\}_{T\in\mathcal{T}}$ of measurement correlations for $(\mathcal{M},S)$,
the following conditions are equivalent:\\
$(1)$ There is a system $\{\widetilde{W}_T\}_{T\in\mathcal{T}}$
of measurement correlations for $({\bf B}(\mathcal{H}),S)$ such that
\begin{equation}
W_T(\overrightarrow{M})=\widetilde{W}_T(\overrightarrow{M})
\end{equation}
for all $T\in\mathcal{T}$ and $\overrightarrow{M}\in \mathcal{M}^{|T|}$.\\
$(2)$ There is a measuring process $\mathbb{M}=(\mathcal{K},\sigma,E,U)$
for $(\mathcal{M},S)$ such that
\begin{equation}\label{SEMPandCP2}
W_T(\overrightarrow{M})=W_T^\mathbb{M}(\overrightarrow{M})
\end{equation}
for all $T\in\mathcal{T}$ and $\overrightarrow{M}\in\mathcal{M}^{|T|}$.
\end{corollary}
The proof of this corollary is obvious by Theorem \ref{OTOC}.
It is not known that how large the set of systems of measurement correlations
for $(\mathcal{M},S)$ satisfying the above equivalent conditions in
the set of systems of measurement correlations for $(\mathcal{M},S)$ at the present time.


Going back to the starting point of quantum measurement theory,
we do not have to rack our brain to resolve the above difficulty.
This is because we should recall that each CP instrument statistically
corresponds to an appratus measuring the system under consideration in the sense of the Davies-Lewis proposal.
In addition, the introduction of systems of measurement correlations was motivated by
the necessity of the counterpart of CP instruments
in the (generalized) Heisenberg picture in order to systematically treat measurement correlations.
Hence it is natural to consider that an instrument $\mathcal{I}$ for $(\mathcal{M},S)$
describing a physically realizable measurement should be defined by a system of measurement correlations
$\{W_T\}_{T\in\mathcal{T}_S}$ for $(\mathcal{M},S)$, i.e.,
\begin{equation}
\mathcal{I}(M,\Delta)=\mathcal{I}_W(M,\Delta)
\end{equation}
for all $M\in\mathcal{M}$ and $\Delta\in\mathcal{F}$.
\begin{question}\label{Q1}
Let $\mathcal{M}$ be a von Neumann algebra on a Hilbert space
$\mathcal{H}$ and $(S,\mathcal{F})$ a measurable space.
For any CP instrument $\mathcal{I}$ for $(\mathcal{M},S)$, does there exist
a system of measurement correlations $\{W_T\}_{T\in\mathcal{T}_S}$ for $(\mathcal{M},S)$
which defines $\mathcal{I}$?
\end{question}
In the case of ${\bf B}(\mathcal{H})$, this question is already affirmatively answered by
the existence of measuring processes for $({\bf B}(\mathcal{H}),S)$
for every CP instrument for $({\bf B}(\mathcal{H}),S)$ (Theorem \ref{CPIMP}).
Surprisingly, Question \ref{Q1} is affirmatively resolved for all CP instruments defined on
arbitrarily given von Neumann algebras.
\begin{theorem}
Let $\mathcal{M}$ be a von Neumann algebra on a Hilbert space $\mathcal{H}$
and $(S,\mathcal{F})$ a measurable space.
For every CP instrument $\mathcal{I}$ for $(\mathcal{M},S)$, 
there exists a system of measurement correlations $\{W_T\}_{T\in\mathcal{T}_S}$ for $(\mathcal{M},S)$
such that
\begin{equation}
\mathcal{I}(M,\Delta)=\mathcal{I}_W(M,\Delta)
\end{equation}
for all $M\in\mathcal{M}$ and $\Delta\in\mathcal{F}$.
\end{theorem}
\begin{proof}
By \cite[Proposition 4.2]{Oz84} (or \cite[Proposition 3.2]{OO16}),
there exist a Hilbert space $\mathcal{K}$, a normal representation
$\pi_0$ of $\mathcal{M}$ on $\mathcal{K}$, a PVM $E_0:\mathcal{F}\rightarrow{\bf B}(\mathcal{K})$
and an isometry $V:\mathcal{H}\rightarrow\mathcal{K}$ such that
\begin{align}
\mathcal{I}(M,\Delta) &= V^\ast \pi_0(M)E_0(\Delta)V, \label{CPrep} \\
\pi_0(M)E_0(\Delta) &= E_0(\Delta)\pi_0(M)
\end{align}
for all $M\in\mathcal{M}$ and $\Delta\in\mathcal{F}$, and that
$\mathcal{K}=\overline{\mathrm{span}}(\pi_0(\mathcal{M})E_0(\mathcal{F})V\mathcal{H})$.

We follow the identification
\begin{equation}
{\bf B}(\mathcal{H}\oplus\mathcal{K})=\left(
\begin{array}{cc}
{\bf B}(\mathcal{H}) & {\bf B}(\mathcal{K},\mathcal{H}) \\
{\bf B}(\mathcal{H},\mathcal{K}) & {\bf B}(\mathcal{K})
\end{array}
\right)
\end{equation}
with multiplication and involution compatible with the usual matrix operations.
We define a normal represetation $\Pi_{in}$ of $\mathcal{M}$ on $\mathcal{H}\oplus\mathcal{K}$ by
\begin{equation}
\Pi_{in}(M)=\left(
\begin{array}{cc}
M & 0 \\
0 & \pi_0(M)
\end{array}
\right) \label{inrepthm5.1}
\end{equation}
for all $M\in\mathcal{M}$, a PVM
$E:\mathcal{F}\rightarrow{\bf B}(\mathcal{H}\oplus\mathcal{K})$ by
\begin{equation}
E(\Delta)=\left(
\begin{array}{cc}
\delta_{s}(\Delta)1 & 0 \\
0 & E_0(\Delta)
\end{array}
\right)
\end{equation}
for all $\Delta\in\mathcal{F}$, where $s\in S$
and $\delta_s$ is a delta measure on $(S,\mathcal{F})$ concentrated on $s$,
and a unitary operator $U$ of $\mathcal{H}\oplus\mathcal{K}$ by
\begin{equation}
U=\left(
\begin{array}{cc}
0 & -V^\ast \\
V & Q
\end{array}
\right),
\end{equation}
where $Q=1-VV^\ast$.
For every $\Delta\in\mathcal{F}$, we define a representation $\Pi_\Delta$
of $\mathcal{M}$ on $\mathcal{H}\oplus\mathcal{K}$ by
\begin{align}
\Pi_{\Delta}(M)&=U^\ast \Pi_{in}(M)E(\Delta)U \nonumber\\
 &=\left(
\begin{array}{cc}
\mathcal{I}(M,\Delta) & -V^\ast\pi_0(M)E_0(\Delta)Q \\
-Q\pi_0(M)E_0(\Delta)V & \delta_s(\Delta)VMV^\ast+Q\pi_0(M)E_0(\Delta)Q
\end{array}
\right) \label{outrepthm5.1}
\end{align}
for all $M\in\mathcal{M}$. We define a unital normal CP linear map
$P_{11}:{\bf B}(\mathcal{H}\oplus\mathcal{K})\rightarrow{\bf B}(\mathcal{H})$ by
\begin{equation}
\left(
\begin{array}{cc}
{\bf B}(\mathcal{H}) & {\bf B}(\mathcal{K},\mathcal{H}) \\
{\bf B}(\mathcal{H},\mathcal{K}) & {\bf B}(\mathcal{K})
\end{array}
\right)\ni
\left(
\begin{array}{cc}
X_{11} &X_{12}   \\
 X_{21} & X_{22} 
\end{array}
\right)\mapsto X_{11}\in {\bf B}(\mathcal{H}).
\end{equation}
For every $T=(t_1,\cdots,t_{|T|})\in\mathcal{T}_S$,
we define a map $W_T:\mathcal{M}^{|T|}\rightarrow {\bf B}(\mathcal{H})$ by
\begin{equation}
W_T(\overrightarrow{M})=P_{11}[\Pi_T(\overrightarrow{M})]=P_{11}[\Pi_{t_1}(M_1)\cdots \Pi_{t_{|T|}}(M_{|T|})]
\end{equation}
for all $\overrightarrow{M}=(M_1,\cdots,M_{|T|})\in \mathcal{M}^{|T|}$.

We show that the family $\{W_T\}_{T\in\mathcal{T}_S}$ is a system of measurement correlations
for $(\mathcal{M},S)$ such that
\begin{equation}
\mathcal{I}(M,\Delta)=W_\Delta(M)
\end{equation}
for all $M\in\mathcal{M}$ and $\Delta\in\mathcal{F}$. For this purpose,
it suffices to show $W_T(\overrightarrow{M})\in\mathcal{M}$ for all $T\in\mathcal{T}_S$ and
$\overrightarrow{M}\in \mathcal{M}^{|T|}$.
Then the set
\begin{equation}
\mathcal{D}=\left(
\begin{array}{cc}
\mathcal{M} & \mathrm{span}(\mathcal{M}V^\ast \mathcal{A}) \\
\mathrm{span}(\mathcal{A}V \mathcal{M}) & \mathcal{A}
\end{array}
\right)
\end{equation}
is a $^\ast$-subalgebra of ${\bf B}(\mathcal{H}\oplus\mathcal{K})$,
where $\mathcal{A}$ is a $^\ast$-subalgebra of ${\bf B}(\mathcal{K})$
algebraically generated by $V\mathcal{M}V^\ast$, $\pi_0(\mathcal{M})E_0(\mathcal{F})$ and $Q=1-VV^\ast$.
This fact follows from the usual matrix operations
and $V^\ast\mathcal{A}V\subset\mathcal{M}$\footnote{To show this, we use $Q=1-VV^\ast$ and Eq.(\ref{CPrep}).}.
Since it is obvious that $\Pi_{in}(M), \Pi_\Delta(M^\prime)\in\mathcal{D}$
for all $M,M^\prime\in\mathcal{M}$ and $\Delta\in\mathcal{F}$,
we have $\Pi_T(\overrightarrow{M})\in\mathcal{D}$ for all $T\in\mathcal{T}_S$ and $\overrightarrow{M}\in \mathcal{M}^{|T|}$.
Therefore,
for every $T\in\mathcal{T}_S$ and $\overrightarrow{M}\in \mathcal{M}^{|T|}$,
the $(1,1)$-component of $\Pi_{T}(\overrightarrow{M})$ is also an element of $\mathcal{M}$, which completes the proof.
\end{proof}

\begin{remark}
In the case of an atomic von Neumann algebra $\mathcal{M}$ on a Hilbert space $\mathcal{H}$,
we have another construction of a system of measurement correlations which defines a given CP instrument.

Let $\mathcal{E}:{\bf B}(\mathcal{H})\rightarrow\mathcal{M}$ be a normal conditional expectation.
We define a CP instrument $\widetilde{\mathcal{I}}$ for $({\bf B}(\mathcal{H}),S)$ by
\begin{equation}
\widetilde{\mathcal{I}}(X,\Delta)=\mathcal{I}(\mathcal{E}(X),\Delta)
\end{equation}
for all $X\in{\bf B}(\mathcal{H})$ and $\Delta\in\mathcal{F}$. By Theorem \ref{CPIMP},
there exists a measuring process $\mathbb{M}=(\mathcal{K},\sigma,E,U)$ for $({\bf B}(\mathcal{H}),S)$
such that 
\begin{equation}
\widetilde{\mathcal{I}}(X,\Delta)=\mathcal{I}_\mathbb{M}(X,\Delta)
\end{equation}
for all $X\in{\bf B}(\mathcal{H})$ and $\Delta\in\mathcal{F}$.
A system of measurement correlations $\{W_T\}_{T\in\mathcal{T}_S}$ for $(\mathcal{M},S)$ is defined by
\begin{equation}
W_T(\overrightarrow{M})=\mathcal{E}(W_T^\mathbb{M}(\overrightarrow{M}))
\end{equation}
for all $T\in\mathcal{T}$ and $\overrightarrow{M}\in\mathcal{M}^{|T|}$. Then $\mathcal{I}_W$ satisfies
\begin{equation}
\mathcal{I}_W(M,\Delta)=\mathcal{E}(W_\Delta^\mathbb{M}(M))=\mathcal{E}(\widetilde{\mathcal{I}}(M,\Delta))
=\mathcal{E}(\mathcal{I}(\mathcal{E}(M),\Delta))=\mathcal{I}(M,\Delta)
\end{equation}
for all $M\in\mathcal{M}$ and $\Delta\in\mathcal{F}$.

We should remark that the above construction does not show
the existence of measuring processes for $(\mathcal{M},S)$
for every CP instrument for $(\mathcal{M},S)$.
\end{remark}

\section{Approximate realization of CP instruments by measuring processes}
\label{sec:6}

We discuss the realizability of CP instruments by measuring processes in this section.
Here, we shall start from the following question similar to Question \ref{Q1}.

\begin{question}\label{Q4}
Let $\mathcal{M}$ be a von Neumann algebra on a Hilbert space
$\mathcal{H}$ and $(S,\mathcal{F})$ a measurable space.
For any CP instrument $\mathcal{I}$ for $(\mathcal{M},S)$, does there exist
a measuring process $\mathbb{M}$ for $(\mathcal{M},S)$ 
which realizes $\mathcal{I}$ within arbitrarily given error limits $\varepsilon>0$?
\end{question}
We say that a CP instrument $\mathcal{I}$ for $(\mathcal{M},S)$ is \textit{approximately realized} by
a net of measuring processes $\{\mathbb{M}_\alpha\}_{\alpha\in A}$ for $(\mathcal{M},S)$, or
$\{\mathbb{M}_\alpha\}_{\alpha\in A}$ \textit{approximately realizes} $\mathcal{I}$ if,
for every $\varepsilon>0$, $n\in\mathbb{N}$, $\rho_1,\cdots,\rho_n\in\mathcal{S}_n(\mathcal{M})$,
$\Delta_1,\cdots,\Delta_n\in\mathcal{F}$ and $M_1,\cdots,M_n\in\mathcal{M}$,
 there is $\alpha\in A$ such that $|\langle \rho_i, \mathcal{I}(M_i,\Delta_i)\rangle-\langle \rho_i,
 \mathcal{I}_{\mathbb{M}_\alpha}(M_i,\Delta_i)\rangle|<\varepsilon$ for all $i=1,\cdots,n$.
$\mathrm{CPInst}_{\mathrm{AR}}(\mathcal{M},S)$ denotes the set of 
CP instruments for $(\mathcal{M},S)$ approximately realized
by nets of measuring processes for $(\mathcal{M},S)$.

Before answering to Question \ref{Q4}, we shall extend the program,
advocated and developed by many researchers \cite{HK64,Sch59,Lud67,Kraus},
which states that physical processes should be described by (inner) CP maps
usually called operations \cite{HK64} or effects \cite{Lud67}.

\begin{definition}[\cite{Mingo89,ADH90}]
Let $\mathcal{M}$ be a von Neumann algebra on a Hilbert space $\mathcal{H}$.\\
$(1)$ A positive linear map $\Psi$ of $\mathcal{M}$ is said to be finitely inner if
there is a finite sequence $\{V_j\}_{j=1,\cdots,m}$ of $\mathcal{M}$ such that
\begin{equation}
\Psi(M)=\sum_{j=1}^m V_j^\ast MV_j
\end{equation}
for all $M\in\mathcal{M}$\\
$(2)$ A positive linear map $\Psi$ of $\mathcal{M}$ is said to be inner if
there is a sequence $\{V_j\}_{j\in\mathbb{N}}$ of $\mathcal{M}$ such that
\begin{equation}
\Psi(M)=\sum_{j=1}^\infty V_j^\ast MV_j
\end{equation}
for all $M\in\mathcal{M}$, where the convergence is ultraweak.\\
$(3)$ A positive linear map $\Psi$ of $\mathcal{M}$ is said to be approximately inner if
it is the pointwise ultraweak limit of a net $\{\Psi_\alpha\}_{\alpha\in A}$
of finitely inner positive linear maps such that $\Psi_\alpha(1)\leq \Psi(1)$ for all $\alpha\in A$.
\end{definition}
In \cite{ADH90}, finite innerness and approximate innerness of CP maps are called factorization
through the identity map $\mathrm{id}_\mathcal{M}:\mathcal{M}\rightarrow\mathcal{M}$ and
approximate factorization through $\mathrm{id}_\mathcal{M}$, respectively.
We refer the reader to \cite{Mingo89,Mingo90,ADH90,AD95} for more detailed discussions.
It is obvious that every finitely inner positive linear map $\Psi$ of $\mathcal{M}$ is inner.
Similarly, every inner positive linear map $\Psi(M)=\sum_{j=1}^\infty V_j^\ast M V_j$,
$M\in\mathcal{M}$, is approximately inner since it is the ultraweak limit of
a sequence $\{\Psi_j\}$ of finitely inner positive maps
$\Psi_j(M)=\sum_{k=1}^j V_k^\ast M V_k$, $M\in\mathcal{M}$,
such that $\Psi_j(1)=\sum_{k=1}^j V_k^\ast V_k\leq \Psi(1)$ for all $j\in\mathbb{N}$.
Every approximately inner positive linear map $\Psi$ of $\mathcal{M}$ is always completely positive.
\begin{definition}
An instrument $\mathcal{I}$ for $(\mathcal{M},S)$ is said to be finitely inner 
[inner, or approximately inner, respectively] if
 $\mathcal{I}(\cdot,\Delta)$ is finitely inner
 [inner, or approximately inner, respectively] for every $\Delta\in\mathcal{F}$.
$\mathrm{CPInst}_{\mathrm{FI}}(\mathcal{M},S)$ 
[$\mathrm{CPInst}_{\mathrm{IN}}(\mathcal{M},S)$, or $\mathrm{CPInst}_{\mathrm{AI}}(\mathcal{M},S)$,
respectively] denotes
the set of finitely inner [inner, or approximately inner, respectively] CP instruments for $(\mathcal{M},S)$.
\end{definition}
The following relation holds.
\begin{equation}
\mathrm{CPInst}_{\mathrm{FI}}(\mathcal{M},S)\subset\mathrm{CPInst}_{\mathrm{IN}}(\mathcal{M},S)
\subset\mathrm{CPInst}_{\mathrm{AI}}(\mathcal{M},S).
\end{equation}

\begin{definition}[Inner measuring process]
A measuring process $\mathbb{M}=(\mathcal{K},\sigma,E,U)$ for $(\mathcal{M},S)$
is said to be inner if $U$ is contained in $\mathcal{M}\overline{\otimes} {\bf B}(\mathcal{K})$.
\end{definition}
We then have the following theorem. 
\begin{theorem} \label{INJINNApp}
Let $\mathcal{M}$ be a von Neumann algebra on a Hilbert space $\mathcal{H}$
and $(S,\mathcal{F})$ a measurable space.
For every approximately inner (hence CP) instrument $\mathcal{I}$ for $(\mathcal{M},S)$, $\varepsilon>0$,
$n\in\mathbb{N}$, $\rho_1,\cdots,\rho_n\in\mathcal{S}_n(\mathcal{M})$,
$M_1,\cdots,M_n\in\mathcal{M}$ and $\Delta_1,\cdots,\Delta_n\in\mathcal{F}$,
there exists an inner measuring process $\mathbb{M}=(\mathcal{K},\sigma,E,U)$ for $(\mathcal{M},S)$
in the sense of Definition \ref{MP2} such that
\begin{equation}
|\langle\rho_j,\mathcal{I}(M_j,\Delta_j)\rangle
-\langle\rho_j, \mathcal{I}_\mathbb{M}(M_j,\Delta_j)\rangle|<\varepsilon
\end{equation}
for all $j=1,\cdots,n$.
\end{theorem}
\begin{corollary}
Let $\mathcal{M}$ be a von Neumann algebra on a Hilbert space $\mathcal{H}$
and $(S,\mathcal{F})$ a measurable space. We have
\begin{equation}
\mathrm{CPInst}_{\mathrm{AI}}(\mathcal{M},S)\subset\mathrm{CPInst}_{\mathrm{AR}}(\mathcal{M},S).
\end{equation}
\end{corollary}

Only for injective factors the following holds as a corollary of the above theorem.
\begin{theorem} \label{INJFac}
Let $\mathcal{M}$ be an injective factor on a Hilbert space $\mathcal{H}$
and $(S,\mathcal{F})$ a measurable space.
For every CP instrument $\mathcal{I}$ for $(\mathcal{M},S)$, $\varepsilon>0$,
$n\in\mathbb{N}$, $\rho_1,\cdots,\rho_n\in\mathcal{S}_n(\mathcal{M})$,
$M_1,\cdots,M_n\in\mathcal{M}$ and $\Delta_1,\cdots,\Delta_n\in\mathcal{F}$,
there exists an inner measuring process $\mathbb{M}=(\mathcal{K},\sigma,E,U)$ for $(\mathcal{M},S)$
in the sense of Definition \ref{MP2} such that
\begin{equation}
|\langle\rho_j,\mathcal{I}(M_j,\Delta_j)\rangle
-\langle\rho_j, \mathcal{I}_\mathbb{M}(M_j,\Delta_j)\rangle|<\varepsilon
\end{equation}
for all $j=1,\cdots,n$.
\end{theorem}
\begin{corollary}
Let $\mathcal{M}$ be an injective factor on a Hilbert space $\mathcal{H}$
and $(S,\mathcal{F})$ a measurable space. Then we have
\begin{equation}
\mathrm{CPInst}_{\mathrm{AI}}(\mathcal{M},S)=\mathrm{CPInst}_{\mathrm{AR}}(\mathcal{M},S)
=\mathrm{CPInst}(\mathcal{M},S).
\end{equation}
\end{corollary}
Theorem \ref{INJFac} is a stronger result than \cite[Theorem 4.4]{OO16} for factors,
so that Question \ref{Q4} is affirmatively resolved for injective factors.
We use the following proposition for the proof of Theorem \ref{INJFac}.
\begin{proposition}[Anantharaman-Delaroche and Havet \text{\cite[Lemma 2.2, Remarks 5.4]{ADH90}}]\label{injinn}
Let $\mathcal{M}$ be an injective factor on on a Hilbert space $\mathcal{H}$.
Every CP map $\Psi$ of $\mathcal{M}$ is approximately inner, i.e.,
it is the pointwise ultraweak limit of
a net of CP maps $\{\Psi_\theta\}_{\theta\in\Theta}$
of the form $\Psi_\theta(M)=\sum_{j=1}^{n_\theta}V_{\theta,j}^\ast MV_{\theta,j}$, $M\in\mathcal{M}$, with
$n_\theta\in\mathbb{N}$, $V_{\theta,1},\cdots, V_{\theta,n_\theta}\in\mathcal{M}$ such that
$\Psi_\theta(1)=\sum_{j=1}^{n_\theta}V_{\theta,j}^\ast V_{\theta,j}\leq \Psi(1)$ for all $\theta\in\Theta$.
\end{proposition}
The following proof is inspired by \cite{PT15}.
\begin{proof}[Proof of Theorem \ref{INJINNApp}]
Let $n\in\mathbb{N}$, $\rho_1,\cdots,\rho_n\in\mathcal{S}_n(\mathcal{M})$,
$M_1,\cdots,M_n\in\mathcal{M}\backslash\{0\}$ and $\Delta_1,\cdots,\Delta_n\in\mathcal{F}\backslash\{\emptyset\}$.
Let $\mathcal{F}^\prime$ be a $\sigma$-subfield of
$\mathcal{F}$ generated by $\Delta_1,\cdots,\Delta_n,S$.
Let $\{\Gamma_i\}_{i=1}^m \subset \mathcal{F}^\prime\backslash\{\emptyset\}$
be a maximal partition of $\cup_{i=1}^n\Delta_i$, i.e.,
$\{\Gamma_i\}_{i=1}^m$ satisfies the following conditions:\\
$(1)$ For every $i=1,\cdots,m$, if $\Delta\in\mathcal{F}^\prime$ satisfies $\Delta\subset\Gamma_i$,
then $\Delta$ is $\Gamma_i$ or $\emptyset$;\\
$(2)$ $\cup_{i=1}^m \Gamma_i=\cup_{i=1}^n\Delta_i$;\\
$(3)$ $\Gamma_i \cap\Gamma_j=\emptyset$ if $i\neq j$.\\
For every $i=1,\cdots,m$, there is a net of finitely inner CP maps $\{\Psi_{\theta_i}\}_{\theta_i\in\Theta_i}$
of the form $\Psi_{i,\theta_i}(M)=
\sum_{j=1}^{n_{\theta_i}}V_{i,\theta_i,j}^\ast MV_{i,\theta_i,j}$, $M\in\mathcal{M}$, with
$n_{\theta_i}\in\mathbb{N}$, $V_{i,\theta_i,1},\cdots, V_{i,\theta_i,n_{\theta_i}}\in\mathcal{M}$
such that is pointwisely convergent to $\mathcal{I}(\cdot,\Gamma_i)$ in the ultraweak topology
and $\Psi_{i,\theta_i}(1)\leq\mathcal{I}(1,\Gamma_i)$.

We fix $s_0,s_1,\cdots.s_m\in S$ such that $s_i\in\Gamma_i$ for all $i=1,\cdots,m$.
For every ${\bm \theta}=(\theta_1,\cdots,\theta_m)$ $\in\Theta_1\times\cdots\times\Theta_m$,
we define a finitely inner CP instrument $\mathcal{I}_{\bm \theta}$ for $(\mathcal{M},S)$ by
\begin{equation}
\mathcal{I}_{\bm \theta}(M,\Delta)=\sum_{i=1}^m \delta_{s_i}(\Delta) \Psi_{\theta_i}(M)
+\delta_{s_m}(\Delta) L_{\bm \theta}ML_{\bm \theta}
\end{equation}
for all $M\in\mathcal{M}$ and $\Delta\in\mathcal{F}$,
where $\delta_s$ is a delta measure on $(S,\mathcal{F})$ concentrated on $s$ and
$L_{\bm \theta}=\sqrt{1-\sum_{i=1}^m \Psi_{i,\theta_i}(1)}$.

Let $\varepsilon$ be a positive real number.
For every $i=1,\cdots,m$, there is $\bar{\theta}_i\in\Theta_i$ such that
\begin{equation}\label{mainterm}
|\langle \rho_j,\mathcal{I}(M_j,\Gamma_i)\rangle-\langle \rho_j,\Psi_{i,\bar{\theta}_i}(M_j)\rangle|
<\frac{\varepsilon}{2m}
\end{equation}
for all $j=1,\cdots,n$, and that
\begin{equation}\label{residual}
\langle \rho_j, \mathcal{I}(1,\Gamma_i)-\Psi_{i,\bar{\theta}_i}(1) \rangle
<\frac{\varepsilon}{2m\sum_{k=1}^n \Vert M_k\Vert}
\end{equation}
for all $j=1,\cdots,n$.

By Eqs.(\ref{mainterm}), (\ref{residual}), we have
\begin{align}
&\quad\; |\langle \rho_j,\mathcal{I}(M_j,\Delta_j)\rangle-
\sum_{i=1}^m\delta_{s_i}(\Delta_j)\langle \rho_j,\Psi_{i,\bar{\theta}_i}(M_j)\rangle| \nonumber\\
 &=| \sum_{i=1}^m\delta_{s_i}(\Delta_j)\langle\rho_j,\mathcal{I}(M_j,\Gamma_i)\rangle-
\sum_{i=1}^m\delta_{s_i}(\Delta_j)\langle \rho_j,\Psi_{i,\bar{\theta}_i}(M_j)\rangle| \nonumber \\
 &\leq \sum_{i=1}^m\delta_{s_i}(\Delta_j)| \langle\rho_j,\mathcal{I}(M_j,\Gamma_i)\rangle
 -\langle \rho_j,\Psi_{i,\bar{\theta}_i}(M_j)\rangle|
 < \sum_{i=1}^m\delta_{s_i}(\Delta_j)\cdot\frac{\varepsilon}{2m} \leq \frac{\varepsilon}{2}
\end{align}
for all $j=1,\cdots,n$, and
\begin{align}
|\rho_j(L_{\bm \bar{\theta}}M_jL_{\bm \bar{\theta}})| 
\leq \Vert M_j\Vert \rho_j(L_{\bm \bar{\theta}}^2)
&=\Vert M_j\Vert\cdot \langle\rho_j,\sum_{i=1}^m \left(\mathcal{I}(1,\Gamma_i)-
 \Psi_{i,\bar{\theta}_i}(1)\right)\rangle \nonumber \\
 &= \Vert M_j\Vert\cdot \sum_{i=1}^m \langle\rho_j,\mathcal{I}(1,\Gamma_i)-
 \Psi_{i,\bar{\theta}_i}(1)\rangle \nonumber \\
 &< \Vert M_j\Vert\cdot m\cdot \frac{\varepsilon}{2m\sum_{k=1}^n \Vert M_k\Vert}\leq\frac{\varepsilon}{2}.
\end{align}
for all $j=1,\cdots,n$.
Then the CP instrument $\mathcal{I}_{\bm \bar{\theta}}$ for $(\mathcal{M},S)$ 
with ${\bm \bar{\theta}}=(\bar{\theta}_1,\cdots,\bar{\theta}_m)$ satisfies
\begin{equation}
|\langle\rho_j,\mathcal{I}(M_j,\Delta_j)\rangle
-\langle\rho_j, \mathcal{I}_{\bm \bar{\theta}}(M_j,\Delta_j)\rangle|<\varepsilon
\end{equation}
for all $j=1,\cdots,n$.

Next, we shall define an inner measuring process
$\mathbb{M}=(\mathcal{K},\sigma,E,U)$ for $(\mathcal{M},S)$ that realizes $\mathcal{I}_{\bm \bar{\theta}}$.
Let $\eta=\{\eta_j\}_{j=0,1,\cdots,\Vert{\bm \bar{\theta}}\Vert+1}$
be a complete orthonormal system of $\mathbb{C}^{\Vert{\bm \bar{\theta}}\Vert+2}$.
A partial isometry $V:\mathcal{H}\otimes \mathbb{C}^{\Vert{\bm \bar{\theta}}\Vert+2}\rightarrow
\mathcal{H}\otimes \mathbb{C}^{\Vert{\bm \bar{\theta}}\Vert+2}$ is defined by
\begin{equation}
V=\sum_{i=1}^m \sum_{j=\sum_{k=1}^{i-1}\bar{\theta}_k+1}^{\sum_{k=1}^{i}\bar{\theta}_k}
V_{i,\bar{\theta}_i,j}\otimes |\eta_j\rangle\langle \eta_0|
+L_{\bm \bar{\theta}}\otimes |\eta_{\Vert{\bm \bar{\theta}}\Vert+1}\rangle\langle \eta_0|
\end{equation}
It is obvious that $V$ satisfies $V^\ast V =1\otimes |\eta_0\rangle\langle \eta_0|$.
We define a PVM $E:\mathcal{F}\rightarrow M_{\Vert{\bm \bar{\theta}}\Vert+2}(\mathbb{C})$ by
\begin{equation}
E_\eta(\Delta)=
\delta_{s_0}(\Delta)|\eta_0\rangle\langle \eta_0|+
\sum_{i=1}^m \delta_{s_i}(\Delta) \sum_{j=\sum_{k=1}^{i-1}\bar{\theta}_k+1}^{\sum_{k=1}^{i}\bar{\theta}_k}
|\eta_j\rangle\langle \eta_j|
+\delta_{s_m}(\Delta)|\eta_{\Vert{\bm \bar{\theta}}\Vert+1}\rangle\langle \eta_{\Vert{\bm \bar{\theta}}\Vert+1}|
\end{equation}
for all $\Delta\in\mathcal{F}$.

We define a Hilbert space $\mathcal{K}=\mathbb{C}^{\Vert{\bm \bar{\theta}}\Vert+2}\otimes\mathbb{C}^2$,
a normal state $\sigma$ on ${\bf B}(\mathcal{K})=M_{\Vert{\bm \bar{\theta}}\Vert+2}(\mathbb{C})
\otimes M_2(\mathbb{C})$,
a PVM $E:\mathcal{F}\rightarrow{\bf B}(\mathcal{K})$
and a unitary operator $U$ on $\mathcal{H}\otimes\mathcal{K}$ by
\begin{align}
\sigma(X)&=\mathrm{Tr}\left[X(|\eta_0\rangle\langle \eta_0|\otimes G_{11})\right],
\hspace{5mm}X\in{\bf B}(\mathcal{K}),\\
E(\Delta) &=E_\eta(\Delta)\otimes 1,\hspace{5mm}\Delta\in\mathcal{F},\\
U =V\otimes G_{11}&+(1-VV^\ast)\otimes G_{12}+(1-V^\ast V)\otimes G_{12}^\ast-V^\ast\otimes G_{22},
\end{align}
respectively, where $\mathrm{Tr}$ is the trace on $M_{\Vert{\bm \bar{\theta}}\Vert+2}(\mathbb{C})
\otimes M_2(\mathbb{C})$ and
\begin{equation}
G_{11}=\left(
\begin{array}{cc}
1 & 0 \\
0 & 0
\end{array}
\right),\hspace{3mm}
G_{12}=\left(
\begin{array}{cc}
0 & 1 \\
0 & 0
\end{array}
\right),\hspace{3mm}
G_{22}=\left(
\begin{array}{cc}
0 & 0 \\
0 & 1
\end{array}
\right).
\end{equation}

Since $U\in\mathcal{M}\overline{\otimes} M_{\Vert{\bm \bar{\theta}}\Vert+2}(\mathbb{C})
\overline{\otimes}M_2(\mathbb{C})$, the 4-tuple
$\mathbb{M}=(\mathcal{K},\sigma,E,U)$ is an inner measuring process for $(\mathcal{M},S)$ satisfying
\begin{equation}
|\langle\rho_j,\mathcal{I}(M_j,\Delta_j)\rangle
-\langle\rho_j, \mathcal{I}_\mathbb{M}(M_j,\Delta_j)\rangle|<\varepsilon
\end{equation}
for all $j=1,\cdots,n$.
\end{proof}

\begin{proof}[Proof of Theorem \ref{INJFac}]
Let $\mathcal{I}$ be a CP instrument for $(\mathcal{M},S)$.
Since $\mathcal{M}$ is an injective factor, $\mathcal{I}(\cdot,\Delta)$ is approximately inner
for every $\Delta\in\mathcal{F}$ by Proposition \ref{injinn}.
Thus the proof of Theorem \ref{INJINNApp} works.
\end{proof}

\begin{remark}
We use the same notations as in the proof of Theorem \ref{INJINNApp}.
In the case where $\mathcal{M}$ is factor,
we have another construction of a measuring process $\mathbb{M}$ for $(\mathcal{M},S)$ such that
$\mathcal{I}_{\bm \bar{\theta}}(M,\Delta)=\mathcal{I}_\mathbb{M}(M,\Delta)$
for all $M\in\mathcal{M}$ and $\Delta\in\mathcal{F}$.

Let $\mathcal{N}$ be an AFD type III factor on a separable Hilbert space $\mathcal{L}$.
Let $Y$ be a partial isometry of $\mathcal{N}$ such that $Y^\ast Y\neq1$ and $YY^\ast\neq 1$.
There then exists a partial isometry $W$ of
$\mathcal{M}\overline{\otimes} M_{\Vert{\bm \bar{\theta}}\Vert+2}(\mathbb{C})\overline{\otimes}\mathcal{N}$
such that $W^\ast W=1-V^\ast V\otimes Y^\ast Y$
$=1-1\otimes |\eta_0\rangle\langle \eta_0|\otimes Y^\ast Y$
and $WW^\ast=1-VV^\ast\otimes YY^\ast$.
We define a unitary operator $U$ of
$\mathcal{M}\overline{\otimes} M_{\Vert{\bm \bar{\theta}}\Vert+2}(\mathbb{C})\overline{\otimes}\mathcal{N}$ by
\begin{equation}
U=V\otimes Y+W,
\end{equation}
and a PVM $E:\mathcal{F}\rightarrow M_{\Vert{\bm \bar{\theta}}\Vert+2}(\mathbb{C})\overline{\otimes}\mathcal{N}$ by
\begin{equation}
E(\Delta) = E_\eta(\Delta)\otimes 1_\mathcal{L}
\end{equation}
for all $\Delta\in\mathcal{F}$.
Let $\psi$ be a unit vector of $\mathcal{L}$ such that $Y^\ast Y\psi=\psi$.
Then we have $W(\xi\otimes\eta_0\otimes \psi)=0$ for all $\xi\in\mathcal{H}$.
We define a normal state $\sigma$ on
$M_{\Vert{\bm \bar{\theta}}\Vert+2}(\mathbb{C})\overline{\otimes}\mathcal{N}$ by
\begin{equation}
\sigma(X)=\langle \eta_0\otimes \psi| X(\eta_0\otimes \psi)\rangle
\end{equation}
for all $X\in M_{\Vert{\bm \bar{\theta}}\Vert+2}(\mathbb{C})\overline{\otimes}\mathcal{N}$.

A Hilbert space $\mathcal{K}$ is then defined by $\mathcal{K}=\mathbb{C}^{\Vert{\bm \bar{\theta}}\Vert+2}
\otimes\mathcal{L}$.
Since $U\in\mathcal{M}\overline{\otimes} M_{\Vert{\bm \bar{\theta}}\Vert+2}(\mathbb{C})
\overline{\otimes}\mathcal{N}$, the 4-tuple
$\mathbb{M}=(\mathcal{K},\sigma,E,U)$ is an inner measuring process for $(\mathcal{M},S)$ satisfying
the desired property.
\end{remark}

Not only for factors $\mathcal{M}$, we have the following theorem
affirmatively resolving Question \ref{Q4} for physically relevant cases.

\begin{definition}
A measuring process $\mathbb{M}=(\mathcal{K},\sigma,E,U)$ for $(\mathcal{M},S)$ is said to be faithful
if there exists a normal faithul representation
$\widetilde{E}:L^\infty(S,\mathcal{I}_\mathbb{M})\rightarrow{\bf B}(\mathcal{K})$ such that
$\widetilde{E}([\chi_\Delta])=E(\Delta)$ for all $\Delta\in\mathcal{F}$.
\end{definition}
This definition is the same as \cite[Definition 3.4]{OO16}
except that the definition of measuring process for $(\mathcal{M},S)$ is different.

\begin{theorem} \label{INJVN}
Let $\mathcal{M}$ be an injective von Neumann algebra on a Hilbert space $\mathcal{H}$
and $(S,\mathcal{F})$ a measurable space.
For every CP instrument $\mathcal{I}$ for $(\mathcal{M},S)$, $\varepsilon>0$,
$n\in\mathbb{N}$, $\rho_1,\cdots,\rho_n\in\mathcal{S}_n(\mathcal{M})$,
$M_1,\cdots,M_n\in\mathcal{M}$ and $\Delta_1,\cdots,\Delta_n\in\mathcal{F}$,
there exists a faithful measuring process $\mathbb{M}=(\mathcal{K},\sigma,E,U)$ for $(\mathcal{M},S)$
in the sense of Definition \ref{MP2} such that
\begin{equation}
|\langle\rho_j,\mathcal{I}(M_j,\Delta_j)\rangle
-\langle\rho_j, \mathcal{I}_\mathbb{M}(M_j,\Delta_j)\rangle|<\varepsilon
\end{equation}
for all $j=1,\cdots,n$, and that
\begin{equation}
\mathcal{I}(1,\Delta)=\mathcal{I}_\mathbb{M}(1,\Delta)
\end{equation}
for all $\Delta\in\mathcal{F}$.
\end{theorem}

\begin{proof}
Suppose that $\mathcal{M}$ is in a standard form without loss of generality.
Then there is a norm one projection $\mathcal{E}:{\bf B}(\mathcal{H})\rightarrow\mathcal{M}$
and a net $\{\Phi_\alpha\}_{\alpha\in A}$ of unital CP maps such that
$\Phi_\alpha(X)\rightarrow^{uw} \mathcal{E}(X)$ for all $X\in{\bf B}(\mathcal{H})$
by \cite[Corollary 3.9]{AD95}(, or \cite[Proposition 4.2]{OO16}).
For every $\alpha\in A$, a CP instrument $\mathcal{I}_\alpha$ for $({\bf B}(\mathcal{H}),S)$ is defined by
\begin{equation}
\mathcal{I}_\alpha(X,\Delta)=\mathcal{I}(\Phi_\alpha(X),\Delta)
\end{equation}
for all $X\in{\bf B}(\mathcal{H})$ and $\Delta\in\mathcal{F}$.
For every $\alpha\in A$, $\mathcal{I}_\alpha$ satisfies
\begin{equation}
\mathcal{I}_\alpha(1,\Delta)=\mathcal{I}(\Phi_\alpha(1),\Delta)=\mathcal{I}(1,\Delta)
\end{equation}
for all $\Delta\in\mathcal{F}$.

Let $\varepsilon>0$,
$n\in\mathbb{N}$, $\rho_1,\cdots,\rho_n\in\mathcal{S}_n(\mathcal{M})$,
$M_1,\cdots,M_n\in\mathcal{M}$ and $\Delta_1,\cdots,\Delta_n\in\mathcal{F}$.
There exists $\alpha_0\in A$ such that
\begin{equation}
|\langle \rho_i, \mathcal{I}(M_i,\Delta_i)\rangle
-\langle \rho_i, \mathcal{I}_{\alpha_0}(M_i,\Delta_i) \rangle|<\varepsilon
\end{equation}
for every $i=1,\cdots,n$.

By \cite[Proposition 3.2]{OO16} and Theorem \ref{vNhom2}, there exist a Hilbert space $\mathcal{L}_1$,
a normal faithful representation $E_1:L^\infty(S,\mathcal{I}_{\alpha_0})
\rightarrow{\bf B}(\mathcal{L}_1)$ and an isometry $V:\mathcal{H}\rightarrow\mathcal{H}\otimes\mathcal{L}_1$
such that
\begin{equation}
\mathcal{I}_{\alpha_0}(X,\Delta)=V^\ast( X\otimes E_1([\chi_\Delta]))V
\end{equation}
for all $X\in{\bf B}(\mathcal{H})$ and $\Delta\in\mathcal{F}$.

Because the discussion below is not needed in the case of $\dim(\mathcal{L}_1)=1$,
we assume that $\dim(\mathcal{L}_1)\geq 2$.
Let $\eta_1$ be a unit vector of $\mathcal{L}_1$.
Let $\mathcal{N}$ be an AFD type III factor on a separable Hilbert space $\mathcal{L}_2$.
We define a partial isometry
$U_1:\mathcal{H}\otimes\mathcal{L}_1\otimes\mathcal{L}_2
\rightarrow\mathcal{H}\otimes\mathcal{L}_1\otimes\mathcal{L}_2$ by
\begin{equation}
U_1(x\otimes\xi\otimes \psi)=\langle \eta_1| \xi\rangle Vx\otimes \psi
\end{equation}
for all $x\in\mathcal{H}$, $\xi\in\mathcal{L}_1$ and $\psi\in\mathcal{L}_2$.
Let $U_2$ be an isometry of ${\bf B}(\mathcal{L}_1)\overline{\otimes}\mathcal{N}$
such that $U_2U_2^\ast = |\eta_1\rangle\langle\eta_1|\otimes 1$.
We define an isometry $U_3$ of ${\bf B}(\mathcal{H}\otimes\mathcal{L}_1)\overline{\otimes}\mathcal{N}$
by $U_3=1\otimes U_2$.
We then define a unitary operator $U$ of $\mathcal{H}\otimes\mathcal{L}_1\otimes
\mathcal{L}_2\otimes\mathbb{C}^2$ by
\begin{align}
U &= U_1U_3\otimes G_{11} + [1-(U_1U_3)
(U_1U_3)^\ast]\otimes G_{12}\nonumber\\
 &\quad\; +[1-(U_1U_3)^\ast(U_1U_3)]\otimes G_{12}^\ast-(U_1U_3)^\ast\otimes G_{22} \nonumber \\
 &= U_1U_3\otimes G_{11} + (1-U_1U_1^\ast)\otimes G_{12}-(U_1U_3)^\ast\otimes G_{22},
\end{align}
where
\begin{equation}
G_{11}=\left(
\begin{array}{cc}
1 & 0 \\
0 & 0
\end{array}
\right),\hspace{3mm}
G_{12}=\left(
\begin{array}{cc}
0 & 1 \\
0 & 0
\end{array}
\right),\hspace{3mm}
G_{22}=\left(
\begin{array}{cc}
0 & 0 \\
0 & 1
\end{array}
\right).
\end{equation}

Let $\eta_2$ be a unit vector of $\mathcal{L}_2$.
We define a Hilbert space $\mathcal{K}=\mathcal{L}_1\otimes\mathcal{L}_2\otimes\mathbb{C}^2$,
a normal state $\sigma$ on ${\bf B}(\mathcal{K})$ by
\begin{equation}
\sigma(X)=\mathrm{Tr}\left[X(|\eta_1\otimes\eta_2\rangle\langle\eta_1\otimes\eta_2|\otimes G_{11})\right]
\end{equation}
for all $X\in{\bf B}(\mathcal{K})$, and a PVM $E:\mathcal{F}\rightarrow{\bf B}(\mathcal{K})$ by
\begin{equation}
E(\Delta) =E_1([\chi_\Delta])\otimes 1_{\mathcal{N}\overline{\otimes}M_2(\mathbb{C})}
\end{equation}
for all $\Delta\in\mathcal{F}$, respectively, where $\mathrm{Tr}$ is the trace on
${\bf B}(\mathcal{L}_1\otimes\mathcal{L}_2)\overline{\otimes} M_2(\mathbb{C})$.

The 4-tuple $\mathbb{M}=(\mathcal{K},\sigma,E,U)$ is then a faithful measuring process for $(\mathcal{M},S)$
that realizes $\mathcal{I}_{\alpha_0}$ and that satisfies
\begin{equation}
|\langle\rho_j,\mathcal{I}(M_j,\Delta_j)\rangle
-\langle\rho_j, \mathcal{I}_\mathbb{M}(M_j,\Delta_j)\rangle|<\varepsilon
\end{equation}
for all $j=1,\cdots,n$, and 
\begin{equation}
\mathcal{I}(1,\Delta)=\mathcal{I}_\mathbb{M}(1,\Delta)
\end{equation}
for all $\Delta\in\mathcal{F}$.
\end{proof}
By the proof of Theorem \ref{INJVN} and facts in Section \ref{sec:2}, we have the following corollaries.
\begin{corollary}\label{RENE}
Let $\mathcal{M}$ be a von Neumann algebra on a Hilbert space $\mathcal{H}$
and $(S,\mathcal{F})$ a measurable space. Then we have
\begin{equation}
\mathrm{CPInst}_{\mathrm{RE}}(\mathcal{M},S)=\mathrm{CPInst}_{\mathrm{NE}}(\mathcal{M},S).
\end{equation}
\end{corollary}
\begin{proof}
Use \cite[Theorem 3.4 (iii)]{OO16}.
\end{proof}
\begin{corollary}
Let $\mathcal{M}$ be an atomic von Neumann algebra on a Hilbert space $\mathcal{H}$
and $(S,\mathcal{F})$ a measurable space. Then we have
\begin{equation}
\mathrm{CPInst}_{\mathrm{RE}}(\mathcal{M},S)=\mathrm{CPInst}(\mathcal{M},S).
\end{equation}
\end{corollary}

\begin{corollary}
Let $\mathcal{M}$ be a von Neumann algebra on a Hilbert space $\mathcal{H}$
and $(S,\mathcal{F})$ a measurable space. Then we have
\begin{equation}
\mathrm{CPInst}_{\mathrm{AR}}(\mathcal{M},S)=\mathrm{CPInst}(\mathcal{M},S).
\end{equation}
\end{corollary}

Following these results, Question \ref{Q4} is affirmatively resolved for general
$\sigma$-finite von Neumann algebras.

Throughout the present paper, we have developed
the dilation theory of systems of measurement correlations and CP instruments,
and established many unitary dilation theorems of them.
In the succeeding paper, we systematically develop
successive and continuous measurements in the generalized Heisenberg picture.
The author believes that the approach to quantum measurement theory
given in the present and succeeding papers contributes to
the categorical (re-)formulation of quantum theory.
On the other hand, though we do not know how it is related to the topic of the paper at the present time,
the future task is to find the connection
with the results of Haagerup and Musat \cite{HaagerupMusat11,HaagerupMusat15},
which develop the asymptotic factorizability of CP maps on finite von Neumann algebras.

\begin{acknowledgement}
The author would like to thank Professor Masanao Ozawa for his useful comments and warm encouragement.
This work was supported by the John Templeton Foundations, No. 35771 and
by the JSPS KAKENHI, No. 26247016, and No. 16K17641.
\end{acknowledgement}
\end{document}